\documentclass[11pt]{amsart}
	\usepackage{amsmath}
	\usepackage{amsthm}
	\usepackage{amssymb}
	\usepackage{version}
	\usepackage[small]{complexity}
	\usepackage{fullpage}
	\usepackage{stmaryrd}
	\usepackage[breaklinks]{hyperref}

	\numberwithin{equation}{section}
	\theoremstyle{plain}
	\newtheorem{thm}[equation]{Theorem}\newtheorem*{thm*}{Theorem}
	\newtheorem{cor}[equation]{Corollary}\newtheorem*{cor*}{Corollary}
	\newtheorem{lem}[equation]{Lemma}\newtheorem*{lem*}{Lemma}
	\newtheorem{prop}[equation]{Proposition}\newtheorem*{prop*}{Proposition}
	\newtheorem*{clm*}{Claim}
	\newtheorem*{conj*}{Conjecture}
	\theoremstyle{definition}
	\newtheorem{defn}[equation]{Definition}\newtheorem*{defn*}{Definition}
	\newtheorem*{ex*}{Example}
	\newtheorem*{notn*}{Notation}
	\theoremstyle{remark}
	\newtheorem{rmk}[equation]{Remark}\newtheorem*{rmk*}{Remark}

	\newcommand{\svert}{\ensuremath{|}}
	\newcommand{\lib}{\ensuremath{\llbracket}}
	\newcommand{\rib}{\ensuremath{\rrbracket}}

	\DeclareMathOperator{\rank}{rank}
	\DeclareMathOperator{\mrank}{m-rank}
	\DeclareMathOperator{\chara}{char}
	\DeclareMathOperator{\tr}{tr}

	\title{Tensor Rank: Some Lower and Upper Bounds}
	\author{Boris Alexeev}\thanks{Email: balexeev@math.princeton.edu,
	Department of Mathematics, Princeton University, Fine Hall, Washington
	Road, Princeton, NJ 08544-1000, Supported by an NSF Graduate Research
	Fellowship}
	\author{Michael Forbes}\thanks{Email: miforbes@mit.edu, 
	Department of Electrical Engineering and Computer Science, MIT CSAIL, 32
	Vassar St., Cambridge, MA 02139, Supported by NSF grant
	6919791 and by MIT CSAIL}
	\author{Jacob Tsimerman}\thanks{Email: jtsimerm@math.princeton.edu,
	Department of
	Mathematics, Princeton University, Fine Hall, Washington
	Road, Princeton, NJ 08544-1000}
	\date{2010-12-08}

	\newlang{\TRANK}{TENSOR-RANK}

\begin{document}
\begin{abstract}
	The results of Strassen~\cite{strassen-tensor} and Raz~\cite{raz} show that good enough
	\textit{tensor rank} lower bounds have implications for algebraic circuit/formula lower
	bounds.

	We explore tensor rank lower and upper bounds, focusing on explicit
	tensors. For odd $d$, we construct field-independent explicit 0/1
	tensors $T:[n]^d\to\mathbb{F}$ with rank at least $2n^{\lfloor
	d/2\rfloor}+n-\Theta(d\lg n)$.  This matches (over $\mathbb{F}_2$) or
	improves (all other fields) known lower bounds for
	$d=3$ and improves (over any field) for odd $d>3$.

	We also explore a generalization of permutation matrices, which we denote permutation
	tensors.  We show, by counting, that there exists an order-3 permutation tensor with
	super-linear rank. We also explore a natural class of permutation
	tensors, which we call group tensors. For any group $G$, we
	define the group tensor $T_G^d:G^d\to\mathbb{F}$, by $T_G^d(g_1,\ldots,g_d)=1$ iff
	$g_1\cdots g_d=1_G$.  We give two upper bounds
	for the rank of these tensors.  The first uses representation theory and works over large
	fields $\mathbb{F}$, showing (among other things) that $\rank_\mathbb{F}(T_G^d)\le |G|^{d/2}$.  We
	also show that if this upper bound is tight, then super-linear tensor
	rank lower bounds would follow.  The second upper bound uses
	interpolation and only works for abelian $G$, showing that over any field $\mathbb{F}$ that
	$\rank_\mathbb{F}(T_G^d)\le O(|G|^{1+\lg d}\lg^{d-1}|G|)$.  In either case, this shows that
	many permutation tensors have far from maximal rank, which is very
	different from the matrix case and thus eliminates many natural
	candidates for high tensor rank.

	We also explore monotone tensor rank.  We give explicit 0/1 tensors $T:[n]^d\to\mathbb{F}$
	that have tensor rank at most $dn$ but have monotone tensor rank exactly $n^{d-1}$.  This is
	a nearly optimal separation.
\end{abstract}
\maketitle
\newpage

\section{Introduction}

Most real-world computing treats data as boolean, and thus made of bits.  However, for some
computational problems this viewpoint does not align with algorithm design.  For example, the
determinant is a polynomial, and computing it typically does not require knowledge of the underlying
bit representation and rather treats the inputs as numbers in some field.  In such settings, it is
natural to consider the computation of the determinant as computing a polynomial over the underlying
field, as opposed to computing a boolean function.  

When computing polynomials, just as when computing boolean functions, there are many different
models of computation to choose. The most general is the \textit{algebraic circuit} model.
Specifically, to compute a polynomial $f$ over a field $\mathbb{F}$ in variables $x_1,\ldots,x_n$,
one defines a directed acyclic graph, with exactly $n$ source nodes (each labeled with a distinct
variable), a single sink node (which is thought of as the output), and internal nodes labeled with
either $+$, meaning addition, or $\times$, meaning  multiplication.  Further, each non-leaf is
restricted to have at most two children nodes.  Computation is defined in the natural way: each
non-source node computes the (polynomial) function of its children according to its label, and
source nodes compute the variable they are labeled with.  One can also consider the
\textit{algebraic formula} model, which requires the underlying graph to be a tree.  In both of
these models, we define the \textit{size} of the circuit/formula to be the total number of nodes in the
graph.

Neither the algebraic circuit nor the formula model are well understood, in the sense that while it
can be shown that there exist polynomials which require large circuits for their computation,
no explicit\footnote{A polynomial is said to be \textit{explicit} if the
coefficient of a monomial $\vec{X}^{\vec{\alpha}}$ is computable by algebraic circuits
of size at most $\poly(|\vec{\alpha}|)$.} examples of such polynomials are known.  Indeed, finding such lower bounds for explicit
functions is considered one of the most difficult problems in computational complexity theory.
Several lower bounds are known, such as Strassen's~\cite{strassen-lb} result (using the result of
Baur-Strassen~\cite{baur-strassen}) that the degree $n$
polynomial $\sum_{i=1}^n x_i^n$ requires $\Omega(n\lg n)$ size circuits.  However, no super-linear
size lower bounds are known for constant-degree polynomials.  In the case of formulas,
Kalorkoti~\cite{kalorkoti} proved a quadratic-size lower bound for an explicit function.

One avenue for approaching improvements for both of these models is by proving lower bounds for
\textit{tensor rank}.  A tensor is a generalization of a matrix, and an order-$d$ tensor is defined
as a function $T:[n]^d\to\mathbb{F}$, where $[n]$ denotes the set $\{1,\ldots,n\}$.  A tensor is
\textit{rank one} if it can be factorized as $T(i_1,\ldots,i_d)=\prod_{j=1}^d\vec{v}_j(i_j)$ for
$\vec{v}_j\in\mathbb{F}^n$.  The rank of a tensor is the minimum $r$ such that $T=\sum_{k=1}^r S_k$ for
rank one tensors $S_k$.  It can be seen that an order-2 tensor is a matrix, and the notions of rank
coincide.  It can also be observed that the rank of an $[n]^d$ tensor is always at most $n^{d-1}$, and a
counting-type argument shows that over any field there exist tensors of rank at
least $n^{d-1}/d$.  A tensor is called \textit{explicit} if $T(i_1,\ldots,i_d)$ can be
computed by algebraic circuits of size at most polynomial in $\poly(d\lg n)$,
that is, at most polynomial in the size of the input $(i_1,\ldots,i_d)$.  All
explicit tensors in this paper will also be uniformly explicit.

Interest in tensors arise from their natural correspondence with certain
polynomials.  Consider the sets of variables $\{X_{i,j}\}_{i\in[n],j\in[d]}$.
Given a tensor $T:[n]^d\to\mathbb{F}$, one can define the polynomial
\[f_T(\{X_{i,j}\}_{i\in[n],j\in[d]})=\sum_{i_1,\ldots,i_d\in[n]}T(i_1,\ldots,i_d)\prod_{j=1}^d
X_{i_j,j}\] This connection was used in the following two results. First,
Strassen~\cite{strassen-tensor} showed that
\begin{thm*}[Strassen~\cite{strassen-tensor}, see also~\cite{vzg}]
	For a tensor $T:[n]^3\to\mathbb{F}$, the circuit size complexity of $f_T$ is
	$\Omega(\rank(T))$.
\end{thm*}
Thus, any super-linear lower-bound for order-3 tensor rank gives a super-linear lower bounds for
general arithmetic circuits, even for the constant degree polynomials.   More recently,
Raz~\cite{raz} proved
\begin{thm*}[Raz~\cite{raz}]
	For a family of tensors $T_n:[n]^{d(n)}\to\mathbb{F}$ with
	$\rank(T_n)\ge n^{(1-o(1))d(n)}$ and
	$\omega(1)\le d(n)\le O(\log n/\log\log n)$, the formula-size complexity of $f_{T_n}$
	is super-polynomial.
\end{thm*}
Thus, while Strassen's result cannot be used to prove super-quadratic circuit-size lower bounds
(because of the upper bounds on order-3 tensor rank), Raz's result shows that tensor rank could be
used to prove very strong lower-bounds.  These results motivate a study of
tensor rank as a model of computation in of itself.

\section{Prior Work}

Strassen's connection between order-3 tensor rank and circuit complexity further
established a close connection between tensor rank to what is known as
\textit{bilinear complexity}.  As several important problems, such as matrix
multiplication and polynomial multiplication, are bilinear, one can study their
bilinear complexity, and thus their order-3 tensor rank.  We interpret various
prior results in the language of tensor rank. For the matrix
multiplication (which corresponds to a tensor of size
$[n^2]\times[n^2]\times[n^2]$), Shpilka~\cite{shpilka}
showed that the tensor rank
is at least $3n^2-o(n^2)$ over $\mathbb{F}_2$, and Bl\"{a}ser~\cite{blaser}
earlier showed that over any field the tensor rank is at least
$2.5n^2-\Theta(n)$. For polynomial multiplication (which corresponds to a
tensor of size $[2n-1]\times[n]\times[n]$), Kaminski~\cite{kaminski}
showed that the tensor rank over $\mathbb{F}_q$ is known to be
$(3+1/\Theta(q^3))n-o(n)$ and earlier work by Brown and
Dobkin~\cite{brown-dobkin} showed that over $\mathbb{F}_2$ the tensor rank is at
least $3.52n$.  Lower bounds for these problems seem difficult, in part because strong
upper bounds exist for both matrix multiplication and polynomial multiplication.

This work attempts to prove tensor rank lower bounds for \textit{any} explicit function,
not just problems of prior interest such as matrix or polynomial multiplication.  Previous work in
this realm include that of Ja'Ja'~\cite{jaja} (see their Theorem 3.6), who used the Kronecker theory of
pencils to show tensor rank lower bounds of $1.5n$ for $[n]\times[n]\times[2]$ tensors, for large
fields.  The work was later expanded by Sumi, Miyazaki, and Sakata~\cite{sumi} to smaller fields.
However, in these works the rank is shown to be at most $1.5n$, so seemingly cannot be pushed
further.

It is also worth noting that H{\aa}stad proved~\cite{20100105.1,20100105.2} that
determining if the tensor rank of $T:[n]^3\to\mathbb{F}$ is at most $r$ is
\NP-hard, for $\mathbb{F}$ finite or the rationals (the problem is also known to
be within $\NP$ for finite $\mathbb{F}$, but not known for the rationals).
Implicit in his work is a tensor rank lower bound (for explicit order-3 tensors)
of $4n/3$.  To the best of our knowledge, the hardness of approximating
tensor rank is an open question.  Part of its difficultly is that any
gap-preserving reduction from \NP\ to tensor rank would automatically yield
lower bounds for explicit tensors.

It is also a folklore result (eg. see Raz~\cite{raz}) that one can reshape, or embed, a $n^{\lfloor
d/2\rfloor}\times n^{\lfloor d/2\rfloor}$ size matrix into a order-$d$ tensor,
thus achieving a $n^{\lfloor d/2\rfloor}$ rank lower bound for $[n]^d$ size
tensors.

\section{Our Results}

Our work has several components, each studying different aspects of the tensor
rank problem.  We first give two new methods in proving tensor rank lower
bounds.  In Section~\ref{sect:combtensors}, we detail the first construction,
which proves the best known\footnote{When comparing this result to those listed
in the prior work, it is helpful to note the differences in size of the
tensors, such as comparing $[n^2]^3$ (for matrix multiplication) to our
$[n]^3$. Thus, over $\mathbb{F}_2$, we essentially match Shpilka's $3n^2-o(n^2)$
matrix multiplication result up to low-order terms.} tensor rank lower bound for a tensor of size $[n]\times[n]\times[n]$
(over any field). In particular, using a generalization of Gaussian elimination
we prove
\begin{thm*}[Corollary~\ref{cor:maincor}]
	Let $\mathbb{F}$ be an arbitrary field.  There are explicit $\{0,1\}$-tensors
	$T_n:[n]^3\to\mathbb{F}$ such that $\rank(T_n)=3n-\Theta(\lg n)$.
\end{thm*}
However, our analysis of this construction is exact so no further improvements
can be made.  In Appendix~\ref{sect:algtensors}, we give
a different order-3 tensor construction with a $3n-\Theta(\lg n)$
rank lower bound over $\mathbb{F}_2$ that has no matching upper bound, and leave as
a open question what is the correct rank. In Appendix~\ref{sect:hightensors}, we show how to extend
the order-3 tensor rank lower bounds to yield a lower bound (for odd $d$) of $2n^{\lfloor
d/2\rfloor}+n-\Theta(d\lg n)$ for the tensor rank of an explicit 0/1 size
$[n]^d$ tensor, which is improves by a factor of 2 on the folklore reshaping
lower bound of $n^{\lfloor d/2\rfloor}$.

In Section~\ref{sect:permtensors}, we explore the tensor rank of permutation tensors.  For matrices,
permutation matrices are all full-rank and have a tight connection with the determinant.
Consequently, it is natural to conjecture that a generalization of permutation matrices, which we
call permutation tensors, have high rank.  In particular, using a counting lower bound for Latin
squares, we show that indeed there is a order-3 permutation tensor with super-linear tensor rank
(over finite fields).  

A natural class of permutation tensors are those constructed from groups.  That is, for a finite
group $G$ we can define the \textit{group tensor} $T_G^d:G^d\to\mathbb{F}$, which is a 0/1 tensor defined by
$T_G^d(g_1,\ldots,g_d)=1$ iff $g_1\cdots g_d=1_G$.  It seems natural to conjecture that these
tensors might also have high-rank.  However, using representation theory we can
give a strong upper bound on the rank of any group
tensor (over large fields such as $\mathbb{C}$).  To prove results over any field, we use interpolation methods and
field-transfer results to bound the rank of any group tensor arising from an
abelian group.  In particular, we have the following theorem.
\begin{thm*}[Theorem~\ref{thm:repthyrank}, and
Corollary~\ref{cor:interpolate-abel-group-any-field}]
	Let $G$ be a finite group.  For ``large'' fields $\mathbb{F}$,
	$\rank_\mathbb{F}(T_G^d)\le |G|^{d/2}$. Further, for any field $\mathbb{F}$, if $G$ is abelian then $\rank_\mathbb{F}(T_G^d)\le
	O(|G|^{1+\lg d}\lg^{d-1}|G|)$.
\end{thm*}
In each case, we show that group tensors have rank far from the maximal
$\Theta_d(|G|^{d-1})$, and thus are not good candidates for high tensor rank for
large $d$ (which are needed for Raz's application). We are unable to place non-trivial upper bounds on the
rank of $T_G^d$ for $G$ non-abelian and small $\mathbb{F}$, but it seems natural to conjecture that
strong upper-bounds exist given the above results.  While these results do not
unconditionally imply any circuit lower
bounds, they elucidate differences between tensor rank and matrix rank
 by  proving that group tensors are not a viable candidate of
high-rank tensors.  However, conditioned on the upper bound given
in Theorem~\ref{thm:repthyrank} being tight, we are able to give super-linear
tensor rank lower bounds for explicit order-3 tensors.

Finally, in Section~\ref{sect:monotone} we explore \textit{monotone} tensor rank. Monotone
computation exploits the idea that if a polynomial only uses positive coefficients (over an ordered
field such as $\mathbb{Q}$), then one might try to compute this polynomial only using positive field
elements.  Previous researchers have tried, in various models, to show that such restricted
computation is much more inefficient than unrestricted computation.  Indeed, for general algebraic
circuits Valiant~\cite{valiant-neg} has shown that allowing negative field elements allows for an
exponential improvement in the efficiency of computing certain polynomials.  We continue in this
line of  work.  In particular, we can show the following nearly optimal separation.
\begin{thm*}[Theorem~\ref{thm:monotone}]
	Let $\mathbb{F}$ be any ordered field.  There is a explicit 0/1 tensor
	$T:[n]^d\to\mathbb{F}$ such that $\rank_\mathbb{F}(T)\le dn$, but the monotone rank of $T$ is $n^{d-1}$.
\end{thm*}

\section{Definitions and Notation}

We first define tensors, and give some basic facts about them.  Throughout this
paper, $[n]$ shall denote the set $\{1,\ldots,n\}$, $\llbracket n\rrbracket$
shall denote the set $\{0,\ldots,n-1\}$, and $\lg n$ shall denote the logarithm
of $n$ base 2.  Further, the notation $\lib E\rib$ (the Iverson bracket) will often be used as an indicator
variable for the event $E$, and can be distinguished from $\lib
n\rib$ by context.

\begin{defn}
	A \textbf{tensor} over a field $\mathbb{F}$ is a function
	$T:\prod_{j=1}^d[n_j] \rightarrow \mathbb{F}$.  It is
	said to have order $d$ and size $(n_1,\ldots,n_d)$.  If all
	of the $n_j$ are equal to $n$, then $T$ is said to have size $n^d$.
	$T$ is said to belong to the \textbf{tensor product space}
	$\otimes_{j=1}^d\mathbb{F}^{n_j}$.
\end{defn}

In later sections, the input space of a tensor will sometimes be a group or a
set $\llbracket n\rrbracket$ instead
of the set $[n]$.  Throughout this paper $\mathbb{F}$ shall denote a arbitrary field, the variable $n$ (or
$(n_1,\ldots,n_d)$) shall be reserved for the tensor size, and $d$ shall be reserved for the
order. $\mathbb{F}_q$ will denote the field on $q$ elements. We can now define
the notion of rank for tensors.

\begin{defn}
	A tensor $T:\prod_{j=1}^d [n_i]\to\mathbb{F}$ is \textbf{simple} if for
	$j\in[d]$ there are
	vectors $\vec{v}_j\in\mathbb{F}^{n_j}$
	such that $T=\otimes_{j=1}^d\vec{v}_j$.  That is,
	for all $i_j\in[n_j]$,
	$T(i_1,\ldots,i_d)=\prod_{j=1}^d\vec{v}_j(i_j)$ where
	$\vec{v}_j(i_j)$ denotes the $i_j$-th coordinate of $\vec{v}_j$.
	\label{defn:simpletensor}
\end{defn}

\begin{defn}
	The \textbf{rank} of a tensor $T:\prod_{j=1}^d[n_j]\to\mathbb{F}$, is defined as the minimum
	number of terms in a summation of simple tensors expressing $T$,
	that is,
	\begin{equation*}
		\rank_\mathbb{F}(T)=\min\left\{r:T=\sum_{k=1}^r
		\otimes_{j=1}^d \vec{v}_{j,k}\text{, }
		\vec{v}_{j,k}\in \mathbb{F}^{n_j} \right\}
	\end{equation*}
\end{defn}

Notice that by definition, a non-zero tensor is simple iff it is of rank one.

The next definition shows how identically sized order-$(d-1)$ tensors can be combined into an
order-$d$ tensor.

\begin{defn}
	For $T_1,\ldots,T_{n_d}\in \otimes_{j=1}^{d-1}\mathbb{F}^{n_j}$
	define $T=[T_1\svert \cdots\svert T_{n_d}]$ by the equation
	$T(i_1,\ldots,i_{d-1},i_d)=T_{i_d}(i_1,\ldots,i_{d-1})$. The $T_{i_j}$ are said to be
	the \textbf{layers} of $T$ (along the $d$-th axis). Layers along other axes are
	defined analogously.

	Conversely, given $T\in\otimes_{j=1}^d\mathbb{F}^{n_j}$, define
	\textbf{the $l$-th layer of $T$ (along the $d$-axis)}, sometimes denoted
	$T_l\in\otimes_{j=1}^{d-1}\mathbb{F}^{n_j}$, to be the tensor
	defined by $T_l(i_1,\ldots,i_{d-1})=T(i_1,\ldots,i_{d-1},l)$. 
	\label{defn:layers}
\end{defn}

\section{Combinatorially-defined Tensors}\label{sect:combtensors}

In this section, we construct combinatorially-defined tensors and prove linear lower bounds for
their rank.  To do so, we use the follow fact about tensors, which is proved in
Appendix~\ref{sect:layerreduction}.  For matrices, this can be seen as a
statement about Gaussian elimination.

\begin{cor*}[Iterative Layer Reduction, Corollary~\ref{cor:layerreduction}]
	For layers $S_1,\ldots,S_{n_d}\in
	\mathbb{F}^{n_1}\otimes\cdots\otimes\mathbb{F}^{n_{d-1}}$ with $S_1,\ldots,S_m$
	linearly independent (as vectors in the space $\mathbb{F}^{n_1\cdots n_{d-1}}$),
	there exist constants $c_{i,j}\in\mathbb{F}$, $i\in\{1,\ldots,m\}$,
	$j\in\{m+1,\ldots, n_d\}$, such that
	\begin{equation}
		\rank([S_1|\ldots | S_{n_d}])\ge
		\rank\left(\left[S_{m+1}+\sum_{i=1}^mc_{i,m+1}S_i\left|\ldots \left|
		S_{n_d}+\sum_{i=1}^mc_{i,n_d}S_i\right.\right.\right]\right)+m
	\end{equation}
\end{cor*}

The idea of this section is to construct tensors such that we can apply
Corollary~\ref{cor:layerreduction} as many times as possible.  As mentioned in
Remark~\ref{rmk:layerreductionbarrier}, for a $[n]^d$ tensor, the lemma can be applied at
most $dn$ times, and thus the lower bounds can at best be $dn$.  In general, the lemma may
not be able to be applied this much because the elimination of layers zeroes out too much of
the tensor. However, in this section we construct tensors (for $d=3$) such that we can
almost apply the lemma $dn$ times.  The result is that we give explicit (order 3) 0/1-tensors with
tensor rank exactly $3n-\Theta(\lg n)$ over any field.  To begin, we apply the
above corollary twice, along two different axes, to get the following lemma.  A
full proof of this lemma, along with other claims in this section, can be found
in Appendix~\ref{sect:combtensorsproofs}.

\begin{lem}
	\label{lem:addident}
	Let $A_1,\ldots, A_k$ be $n\times n$ sized $\mathbb{F}$-matrices.
	Let $I_m$ denote the $m\times m$ identity matrix, and $0_m$ denote the $m\times m$
	zero matrix.  Then,
	\begin{equation}
		\label{eq:addidentrankeven}
		\rank\left(\left[
			\begin{matrix}
				I_n & 0_n\\
				0_n & I_n
			\end{matrix}
			\left|
			\begin{matrix}
				0_n & 0_n\\
				A_1 & 0_n
			\end{matrix}
			\right|
			\cdots
			\left|
			\begin{matrix}
				0_n & 0_n\\
				A_k & 0_n
			\end{matrix}
			\right.
			\right]
			\right)
		\ge
		\rank([A_1|\cdots|A_k])+2n
	\end{equation}
	and
	\begin{equation}
		\label{eq:addidentrankodd}
		\rank\left(\left[
			\begin{matrix}
				0   & 0   & 0\\
				I_n & 0_n & 0\\
				0_n & I_n & 0
			\end{matrix}
			\left|
			\begin{matrix}
				0   & 0   & 0\\
				0_n & 0_n & 0\\
				A_1 & 0_n & 0
			\end{matrix}
			\right|
			\cdots
			\left|
			\begin{matrix}
				0   & 0   & 0\\
				0_n & 0_n & 0\\
				A_k & 0_n & 0
			\end{matrix}
			\right.
			\right]
			\right)
		\ge
		\rank([A_1|\cdots|A_k])+2n
	\end{equation}
	where the left-hand side of Equation~\ref{eq:addidentrankodd} expresses the tensor rank of
	a $[2n+1]\times[2n+1]\times [k]$-sized tensor.
\end{lem}

Applying this lemma recursively yields the following construction.

\begin{defn}
	Let $H:\mathbb{N}\rightarrow\mathbb{N}$ denote the Hamming weight function.  That
	is, $H(n)$ is the number of $1$'s in the binary expansion of $n$.
\end{defn}
\begin{thm}
	\label{thm:mainthm}
	For $i\in\{0,\ldots,\lfloor \lg n\rfloor\}$, let $S_{n,i}$ be an $n\times n$ matrix defined
	in the following recursive manner.
	\begin{itemize}
		\item $S_{1,0}=[1]$
		\item For $2n>1$,
			\begin{equation*}
				S_{2n,i}=
				\begin{cases}
					\begin{bmatrix}
						0_n & 0_n\\
						S_{n,i} & 0_n
					\end{bmatrix}
					&
					\text{if } i<\lfloor \lg n\rfloor
					\\
					\begin{bmatrix}
						I_n & 0_n\\
						0_n & I_n
					\end{bmatrix}
					&
					\text{if } i=\lfloor \lg n\rfloor
					\\
				\end{cases}
			\end{equation*}
		\item For $2n+1>1$,
			\begin{equation*}
				S_{2n+1,i}=
				\begin{cases}
					\begin{bmatrix}
						0 & 0 & 0\\
						0_n & 0_n & 0\\
						S_{n,i} & 0_n & 0
					\end{bmatrix}
					&
					\text{if } i<\lfloor \lg n\rfloor
					\\
					\begin{bmatrix}
						0 & 0 & 0\\
						I_n & 0_n & 0\\
						0_n & I_n & 0
					\end{bmatrix}
					&
					\text{if } i=\lfloor \lg n\rfloor
					\\
				\end{cases}
			\end{equation*}
	\end{itemize}
	
	Then, denoting $T_n=[S_{n,0}|\cdots|S_{n,\lfloor \lg n\rfloor}]$,

	\begin{enumerate}
		\item $T_n$ has size $[n]\times [n]\times [\lfloor \lg
		n\rfloor+1]$\label{mainthm:size}.
		\item $\rank(T_n)=2n-2H(n)+1$ \label{mainthm:rank}.
		\item On inputs $n$ and $(i,j,k)\in[n]\times[n]\times[\lfloor\lg n\rfloor+1]$, $T_n(i,j,k)$
		can be computed in polynomial time. That is, in time $O(\polylog(n))$.\label{mainthm:explicit}
	\end{enumerate}
\end{thm}

Another application of Corollary~\ref{cor:layerreduction} (along the one axis it
has not yet been applied) yields the following
claim.

\begin{cor}
	\label{cor:maincor}
	Define $S_{n,i}$ as in Theorem~\ref{thm:mainthm}.  Let $n\in\mathbb{N}$ and restrict
	to $n\ge 2$.
	Then, for $i\in[n]$ define $n\times n$ matrices $S'_{n,i}$ by
	\begin{equation*}
		S'_{n,i}=
			\begin{cases}
				\begin{bmatrix}
					S_{n-1,i-1} & 0 \\
					0 & 0
				\end{bmatrix}
				&\text{if } i\in[\lfloor\lg(n-1)\rfloor +1]
				\\
				\begin{bmatrix}
					0_{n-1} & \vec{e}_{i-(\lfloor\lg(n-1)\rfloor +1)}\\
					0 & 0
				\end{bmatrix}
				&\text{else}
			\end{cases}
	\end{equation*}
	where $\vec{e}_j\in\mathbb{F}^{n-1}$ is the indicator column vector where
	$\vec{e}_j(k)=\lib j=k\rib $. (Notice that $\lfloor \lg(n-1)\rfloor +1\le n-1$ for all $n\ge
	2$.)  Then, denoting $T'_n=[S'_{n,1}|\cdots|S'_{n,n}]$,
	\begin{enumerate}
		\item $T'_n$ has size $[n]^3$.\label{maincor:size}
		\item $\rank(T'_n)=3n-2H(n-1)-\lfloor \lg (n-1)\rfloor-2\ge 3n-\Theta(\lg n)$.
			\label{maincor:rank}
		\item On inputs $n$ and $(i,j,k)\in[n]\times[n]\times[n]$,
		$T'_n(i,j,k)$ can be computed in polynomial time, that is, 
		$O(\polylog(n))$.\label{maincor:explicit}
	\end{enumerate}
\end{cor}

Note that this analysis is exact.  Appendix~\ref{sect:algtensors} contains
similar lower bounds over $\mathbb{F}_2$ using different methods, where no
non-trivial upper bound is known.

\section{Permutation Tensors}\label{sect:permtensors}

One of the most natural families of full-rank matrices are permutation matrices.
This section examines a natural generalization of permutation matrices to
tensors, which we call permutation tensors.  A counting argument
(Proposition~\ref{prop:countinglb}) shows that
there exists order-3 permutation tensors of super-linear rank (over any fixed finite
field), so it is natural to conjecture that permutation tensors may all have
near-maximal rank, just as in the matrix setting.  However, we show
(Subsection~\ref{subsect:rankub})
that this is false: we give tensor rank upper bounds proving that permutation
tensors constructed from groups have rank much less than maximal.

We begin with the formal definition of permutation tensors.

\begin{defn}
	Let $\mathbb{F}$ be a field, and $T$ be a tensor $T:[n]^d\to\mathbb{F}$.
	$T$ is a \textbf{permutation tensor} if $T$ assumes only 0/1 values,
	and $T$ has exactly one 1 in each generalized row. (A
	\textbf{generalized row},
	sometimes just ``row'', is
	the set of $n$ inputs to $T$ resulting from fixing $d-1$ of the
	coordinates, and varying the remaining coordinates).
\end{defn}

It is not hard to see that order-2 permutation tensors and permutation matrices;
as permutation matrices are those 0/1-matrices such that each row and column have exactly one 1.  

\subsection{Permutation Tensors: Rank Lower Bounds}

We now show that there exist permutation tensors of super-linear rank (over
finite fields).

\begin{prop}
	Let $\mathbb{F}$ be a finite field.  Then there exists a permutation tensor
	$T:[n]^3\to\mathbb{F}$ of rank at least
	$\Omega(n\log_{|\mathbb{F}|}n)$.
	\label{prop:countinglb}
\end{prop}
\begin{proof}
	A \textit{Latin square} is an $n\times n$ matrix, with each entry
	labeled with a symbol from $[n]$, such that no symbol is duplicated in
	any row or column.  Observe that order-3 permutation tensors exactly
	correspond to Latin squares.  
	We now use the following fact about Latin squares, whose proof uses
	lower bounds for the permanent of doubly-stochastic matrices.  
	\begin{thm*}[\cite{vl-w}]
		The number of $n\times n$ Latin squares is at least
		$(n!)^{2n}/n^{n^2}$.
	\end{thm*}
	A standard counting argument completes the claim.
\end{proof}

It remains unclear if this result generalizes to higher orders.  That is, 
can one show that for any $k>3$ there exist permutation tensors of rank at least
$\omega(n^{\lfloor d/2\rfloor})$?

\subsection{Permutation Tensors: Rank Upper Bounds}\label{subsect:rankub}

In this section we define a class of permutation tensors constructed from
finite groups, and show that these tensors have rank far from maximal.  We will give
two rank upper-bound methods.  The first method uses representation theory and
accordingly only works where the group has a complete set of irreducible
representations (which usually means ``large'' fields).  The second method is
based on polynomial interpolation, and while it gives worse upper bounds and
only works for finite abelian groups, it gives results over any field.  Neither of
these methods applies to all finite non-abelian groups over small fields, and the rank
of the corresponding tensors is unclear.

\begin{defn}
	Let $G$ be a finite group (written multiplicatively, with identity $1_G$) , and
	$\mathbb{F}$ a field.  Define the order-$k$ \textbf{group tensor}
	$T_G^d:G^d\to\mathbb{F}$ by \[T_G^d(g_1,\ldots,g_d)=\lib g_1\cdots
	g_d=1_G\rib\]
\end{defn}

We first explore the representation-theory based upper bound.  To do so, we
first cite relevant facts from representation theory.

\begin{thm}[\cite{serre}]
	Let $G$ be a finite group and $\mathbb{F}$ a field.  A
	\textbf{representation} of $G$ is a homomorphism
	$\rho:G\to\mathbb{F}^{d\times d}$, where $d$ is the \textbf{dimension}
	of the representation and is denoted $\dim\rho$. The \textbf{character
	of a representation} $\rho$ is a map $\chi_\rho:G\to\mathbb{F}$ defined by
	$\tr\circ\rho$, that is, taking the trace of the resulting matrix of the
	representation.  

	If $\chara(\mathbb{F})$ is coprime to $|G|$, and $\mathbb{F}$ contains
	$N$-th roots of unity, for $N$ equal to the least common multiple of
	all of the orders of elements of $G$, then there exists a
	\textbf{complete set of irreducible representations}.  In particular,
	for $c$ denoting the number of conjugacy classes of $G$, there is
	a set of representations $\rho_1,\ldots,\rho_c$ and associated
	characters such that (among other properties) we have
	\begin{enumerate}
		\item $\frac{1}{|G|}\sum_{i=1}^c (\dim\rho_i)\cdot\chi_i(g) = \lib g=e\rib$
		\item $\sum_{i=1}^{c} (\dim\rho_i)^2=n$
		\item $\dim\rho_i$ divides $|G|$
	\end{enumerate}

	In particular, for finite abelian groups, $c=n$ and $\dim\rho_i=1$ for
	all $i$.
	\label{thm:repthy}
\end{thm}

Notice that property (1) in the above theorem is an instance of the
column orthonormality relations of character tables, which follow from
the more commonly mentioned row orthonormality relations. 
We now use these facts to derive upper bounds on the rank of $T_G^k$ when the
conditions to the above theorem hold.

\begin{thm}
	Let $G$ be a finite group, $d\ge 2$ and $\mathbb{F}$ a field, such that
	$\chara(\mathbb{F})$ is coprime to $|G|$, and $\mathbb{F}$ contains
	$N$-th roots of unity, for $N$ equal to the least common multiple of the
	orders of elements of $G$.  Then given the irreducible representations
	$\rho_1,\ldots,\rho_c$ for $G$ over $\mathbb{F}$, the order-$d$ group tensor has
	$|G|\le \rank_\mathbb{F}(T_G^d)\le\sum_{i=1}^c(\dim\rho_i)^d\le |G|^{d/2}$.
	
	In particular, for finite abelian groups, $\rank_\mathbb{F}(T_G^k)=|G|$.
	\label{thm:repthyrank}
\end{thm}
\begin{proof}
	
	\underline{$\rank_\mathbb{F}(T_G^d)\ge |G|$:} This follows from
	observing that for fixed $g_3,\ldots,g_d$,
	$T_G^d(\cdot,\cdot,g_3,\ldots,g_d)$ is a permutation matrix, and thus
	its rank (of $|G|$) lower bounds the rank of $T_G^d$ (over any field).
	This can also be seen by induction on Corollary~\ref{cor:droplayers}.
	
	\underline{$G$ abelian $\implies$ $\rank_\mathbb{F}(T_G^k)\le|G|$:}
	Theorem~\ref{thm:repthy} further implies $\dim\rho_i=1$ for all
	irreducible representations of finite abelian groups, which implies
	$\rank_\mathbb{F}(T_G^d)\le |G|$ for abelian groups.  

	\underline{$\sum_{i=1}^c(\dim\rho_i)^d\le |G|^{d/2}$:}
	Theorem~\ref{thm:repthy}(2) shows that $\sum_{i=1}^c(\dim\rho_i)^2=n$.
	Thus the claim is equivalent to showing that for $d\in\mathbb{Z}$, $d\ge
	2$ and $n_i\in\mathbb{R}_{\ge 0}$, $\sum n_i=n\implies \sum n_i^{d/2}\le
	n^{d/2}$.  To show this, we first show that
	$(n+m)^{d/2}+0^{d/2}=(n+m)^{d/2}\ge n^{d/2}+m^{d/2}$.  To see, this,
	observe that assuming without loss of generality that $n\ge m$, we have
	that \[(n+m)^d\ge n^d+\binom{d}{\lceil d/2\rceil}n^{\lceil
	d/2\rceil}m^{\lfloor d/2\rfloor}+m^d\ge
	n^d+2n^{d/2}m^{d/2}+m^d=(n^{d/2}+m^{d/2})^2\] where we use that
	$\binom{d}{\lceil d/2\rceil}\ge d\ge 2$.  Taking square roots yields
	$(n+m)^{d/2}+0^{d/2}\ge n^{d/2}+m^{d/2}$.  Thus, given non-negative
	$n_i$ summing to $n$, one can iteratively zero out certain $n_i$ while
	increasing the sum $\sum n_i^{d/2}$, until only $n_1=n$ and thus $\sum
	n_i^{d/2}=n^{d/2}$.  Thus, this is a bound on the initial sum of $\sum
	n_i^{d/2}$.
	
	\underline{$\rank_\mathbb{F}(T_G^d)\le\sum_{i=1}^c(\dim\rho_i)^d$:}
	The result will follow by constructing, for each $i$, the order-$d$
	tensor \[ T_{\rho_i}^d(g_1,\ldots,g_d)=\chi_i(g_1\cdots g_d)\] in rank
	$(\dim\rho_i)^d$.  Theorem~\ref{thm:repthy}(1) shows that
	$T_G^d=\frac{1}{|G|}\sum_{i=1}^c(\dim\rho_i)\cdot T_{\rho_i}^d$ and so distributing the
	$\dim\rho_i/|G|$ term inside the simple tensors yields the result (where we
	crucially use the restriction on the field characteristic).

	Thus, all that remains is to show that
	$\rank_\mathbb{F}(T_{\rho_i}^d)\le (\dim\rho_i)^d$.  Using the group
	homomorphism properties of the representations and expanding the
	definition of the trace through the matrix multiplication we see
	\begin{align*}
		T_{\rho_i}^d(g_1,\ldots,g_d)
			&= \sum_{k_1=1}^{\dim\rho_i}\cdots\sum_{k_{d}=1}^{\dim\rho_i} (\rho_i(g_1))_{k_1,k_2}\cdots (\rho_i(g_{d-1}))_{k_{d-1},k_{d}}(\rho_i(g_d))_{k_{d},k_1}
	\end{align*}
	and one can observe that for fixed $k_1,\ldots,k_d$, the function
	$(\rho_i(g_1))_{k_1,k_2}\cdots (\rho_i(g_{d-1}))_{k_{d-1},k_{d}}(\rho_i(g_d))_{k_{d},k_1}$ is a simple tensor so
	the above shows $\rank_\mathbb{F}(T_{\rho_i}^d)\le (\dim\rho_i)^d$ as
	desired.
\end{proof}

The above result is possibly tight, motivating the question: is there a
group $G$ and irreducible representation $\rho$ of $G$ such that
$\rank_{\mathbb{F}}(T_{\rho}^d)<(\dim\rho)^d$?  As the above result is tight for
abelian groups, any affirmative answer to the above question would involve a
non-abelian $G$.

Even supposing the above result was tight, one can ask what implications this gives
for circuit lower bounds, especially because group tensors are explicit when the
defining group operation is efficiently computable. However, applying tensor
rank lower bounds to Raz's~\cite{raz} result requires order-$d$ tensors of rank
$n^{(1-o(1))d}$, and Theorem~\ref{thm:repthyrank} shows that no group tensor can
achieve this rank over large fields.  In particular, for the purposes of tensor
rank lower bounds, the lower bounds of Corollary~\ref{cor:maincor-highdim} are
asymptotically (in $d$) as good as the rank achievable by any group tensor (over
large fields).

However, if tight, Theorem~\ref{thm:repthyrank} would yield better lower
bounds than Corollary~\ref{cor:maincor-highdim} for odd $d$.  In particular, the
symmetric group $\mathfrak{G}_n$ has a complete set of irreducible
representations over the rationals~\cite{serre}.  Thus, the tightness of
Theorem~\ref{thm:repthyrank} would imply a lower bound for
$\rank_\mathbb{Q}(T_{\mathfrak{G}_n}^d)$, which is an explicit tensor.  To
understand this lower bound, the following fact is useful.
\begin{thm}[\cite{sym-group-max-rep}]
	The largest dimension of an irreducible representation of
	$\mathfrak{G}_n$ over $\mathbb{Q}$ is of size
	$\sqrt{n!}/e^{\Theta(\sqrt{n})}$
	\label{thm:sym-group-max-rep}
\end{thm}
In particular, for $d=3$, all of the above imply that
$\rank_\mathbb{Q}(T_{\mathfrak{G}_n}^d)\ge
|\mathfrak{G}_n|^{1.5}/e^{\Theta(\sqrt{\log|\mathfrak{G}_n|})}$, which is
$\Omega(n^{1.5-\epsilon})$, for any $\epsilon>0$.  Then, applying
Strassen's~\cite{strassen-tensor} result would yield $\Omega(n^{1.5-\epsilon})$
lower bounds for the (unrestricted) circuit size of explicit degree-3
polynomials (that have 0/1 coefficients).  Such a conclusion would surpass the
best known circuit size lower bound even for super-constant degree polynomials,
which is Strassen's~\cite{strassen-lb} $\Omega(n\log n)$ lower bound for degree
$n$ polynomials.  Thus, tightness of Theorem~\ref{thm:repthyrank} would have
interesting consequences.

Regardless of whether the result is tight, Theorem~\ref{thm:repthyrank} only
works over ``large fields'' in general.  In
particular, it does not (in general) give insight into the rank of group tensors over
fixed finite fields, or even over the rationals.  To take an example, the cyclic
group $\mathbb{Z}_n$ requires $n$-th roots of unity for its irreducible
representations.  While Lemma~\ref{lem:changefields} does show a relation
between $\rank_{\mathbb{Q}}(T_{\mathbb{Z}_n}^d)$ and
$\rank_{\mathbb{Q}[x]/\langle x^n-1\rangle}(T_{\mathbb{Z}_n}^d)$ (where
$\mathbb{Q}[x]/\langle x^n-1\rangle$ is the field of rationals adjoined with a $n$-th
primitive root of unity, so Theorem~\ref{thm:repthyrank} applies) this
relationship implies nothing beyond trivial rank upper bounds.  Thus, to achieve
rank upper bounds for group tensors over small fields we take a different
approach, one using polynomial interpolation.  Our result only applies to finite
abelian groups, but is able to show that have ``low'' rank in this regime.  

\begin{prop}
	Let $\mathbb{F}$ be a field with at least $d(n-1)+1$ elements.  Let
	$T:\llbracket n \rrbracket^d\to\mathbb{F}$ be a tensor such that
	\begin{equation*}
		T(i_1,\ldots,i_d)=\sum_{m=0}^{d(n-1)} c_m \lib i_1+i_2+\cdots+i_d=m\rib
	\end{equation*}
	for constants $c_m\in\mathbb{F}$.  Then, $\rank(T)\le d(n-1)+1$.
	\label{prop:interpolate}
\end{prop}
\begin{proof}[Proof Sketch]
	We sketch the proof here, the full proof is in
	Appendix~\ref{sect:permtensorsproofs}.  The proof follows the 
	result of Ben-Or (as reported in Shpilka-Wigderson~\cite{SW01}) on
	computing the symmetric polynomials efficiently over large fields.  To
	compute a desired polynomial
	$f(\vec{x})$, one can introduce a new variable $\alpha$ and an auxiliary
	polynomial $P(\alpha,\vec{x})$ such that
	\begin{itemize}
		\item $P$ is efficiently computable, of degree at most $d'$ in $\alpha$
		\item For some $m$, $f(\vec{x})=C_{\alpha^m}(P(\alpha,\vec{x}))$. That is, $f$ equals the
		coefficient of $\alpha^m$ in $P$.
	\end{itemize}
	To compute $f$ on input $\vec{x}$, we can then evaluate $P$ on
	$(\alpha_1,\vec{x}),\ldots,(\alpha_{d'+1},\vec{x})$ and then use interpolation to recover
	$C_{\alpha^m}=f(x)$.

	To apply this idea to tensors, we observe the coefficients (in the
	variable $\alpha$) of the
	polynomial 
	\begin{equation*}
		P(\alpha,\{X^{(i)}_j\}_{i,j}):=\prod_{i=1}^d \left(X^{(i)}_0+\alpha X^{(i)}_1+\alpha^2 X^{(i)}_2+\cdots+\alpha^{n-1} X^{(i)}_{n-1}\right)
	\end{equation*}
	exactly correspond to the type of tensors we are trying to produce.  As $P$ is degree at
	most $d(n-1)$ in $\alpha$, and further $P$ is a rank one tensor in disguise, interpolation
	completes the
	result.
\end{proof}

We now turn to using Proposition~\ref{prop:interpolate} to upper bound the rank group tensors formed
from cyclic groups.

\begin{cor}
	\label{cor:interpolate-cyclic}
	Let $\mathbb{F}$ be a field with at least $d(n-1)+1$ elements. Then,
	$\rank(T_{\mathbb{Z}_n}^d)\le d(n-1)+1$.
\end{cor}

Using the Structure Theorem of Abelian Groups the following can now be shown
(for proof see Appendix~\ref{sect:permtensorsproofs}).

\begin{cor}
	\label{cor:interpolate-abel-group}
	Let $G$ be a finite abelian group, and $\mathbb{F}$ be a field with at
	least $|G|$ elements.  Then $\rank_\mathbb{F}(T_G^d)\le |G|^{1+\lg d}$.
\end{cor}

While all of the results based on Proposition~\ref{prop:interpolate} do not
require the field to have large roots of unity, they still require the field to
have large size.  Thus, they seemingly do not answer the question of the rank of
group tensors over small fields.  However, as the next lemma shows (with proof
in Appendix~\ref{sect:permtensorsproofs}), one can
transfer results over large-sized fields to small-sized fields with a minor
overhead.

\begin{lem}
	Let $\mathbb{K}$ a field that extends $\mathbb{F}$.  Then for any tensor
	$T:[n]^d\to\mathbb{F}$, $\rank_\mathbb{F}(T)\le
	(\dim_\mathbb{F}\mathbb{K})^{d-1}\cdot\rank_\mathbb{K}(T)$, where
	$\dim_\mathbb{F}\mathbb{K}$ is the dimension of $\mathbb{K}$ as an
	$\mathbb{F}$-vector space.
	\label{lem:changefields}
\end{lem}

With this field-transfer result, we can now state rank upper bounds for group
tensors (for finite abelian groups) for any field.

\begin{cor}
	\label{cor:interpolate-abel-group-any-field}
	Let $\mathbb{F}$ be any field, and $G$ be a finite abelian group.  Then
	$\rank_\mathbb{F}(T_G^d)\le |G|^{1+\lg d}\lceil\lg|G|\rceil^{d-1}$.  

	In particular, if $G$ is cyclic, then $\rank_{\mathbb{F}}(T_G^d)\le
	d|G|\lceil\lg|G|\rceil^{d-1}$.
\end{cor}

This last result shows that for any finite abelian group, any field and any
large $d$, the rank of the corresponding group tensor is far from possible
$\Omega(n^{d-1})$.  These results do not settle the rank of group tensors for
non-abelian groups over small fields, and leaves the open question whether the
methods of Theorem~\ref{thm:repthyrank} or Proposition~\ref{prop:interpolate}
(or other methods) can resolve this case.

\section{Monotone Tensor Rank}\label{sect:monotone}

We now explore a restricted notion of tensor rank, that of monotone tensor rank.  In algebraic
models of computation, monotone computation requires that the underlying field is ordered, which we
now define.

\begin{defn}
	Let $\mathbb{F}$ be a field.  $\mathbb{F}$ is \textbf{ordered} if there is a linear order
	$<$ such that
	\begin{itemize}
		\item For all $x,y,z\in\mathbb{F}$, $x<y\implies x+z<x+y$.
		\item For all $x,y\in\mathbb{F}$ and $z\in\mathbb{F}_{>0}$, $x<y\implies xz<yz$.
	\end{itemize}
	where $\mathbb{F}_{>0}=\{x|x\in\mathbb{F},x>0\}$.
\end{defn}

Recall that every ordered field has characteristic zero, and thus is infinite.  

Over ordered fields,
computation of polynomials that only use positive coefficients can be done using only positive field
constants, but many works (such as \cite{valiant-neg}) have shown that the circuit model of computation,
the restriction to positive field constants in computation leads drastically worse efficiency as
compared to unrestricted computation.  In this section, we show that in the tensor rank model of
computation, monotone computation is also much less efficient then unrestricted computation.  We
first define the notion of monotone tensor rank.

\begin{defn}
	Let $\mathbb{F}$ be a ordered field. Consider a tensor
	$T:\prod_{i=1}^d[n_i]\to\mathbb{F}_{\ge 0}$.  Define the
	\textbf{monotone tensor rank} of $T$, denoted $\mrank(T)$, to be
	\begin{equation*}
		\mrank(T)=\min\left\{r:T=\sum_{l=1}^r
		\vec{v}_{l,1}\otimes\cdots\otimes \vec{v}_{l,d}\text{, }
		\vec{v}_{l,i}\in (\mathbb{F}_{\ge 0})^{n_i} \right\}
	\end{equation*}
\end{defn}

We now show an essentially maximal separation between monotone tensor rank and unrestricted tensor
rank, for the explicit group tensor $T_{\mathbb{Z}_n}^d$.

\begin{thm}
	\label{thm:monotone}
	Let $\mathbb{F}$ be an ordered field.  Consider the group tensor $T_{\mathbb{Z}_n}^d$.  Then
	\begin{enumerate}
		\item $\rank_\mathbb{F}(T_{\mathbb{Z}_n}^d)\le d(n-1)+1$
		\item $\mrank_\mathbb{F}(T_{\mathbb{Z}_n}^d)=n^{d-1}$
	\end{enumerate}
\end{thm}
\begin{proof}[Proof Sketch, see Appendix~\ref{sect:monotoneproofs} for a full
proof]

	The upper bounds follow from Corollary~\ref{cor:interpolate-cyclic}, and
	from the trivial $n^{d-1}$ upper bound for tensor rank.

	The lower bounds follows from the observation that a non-negative simple
	tensor ``covers'' non-zero entries in $T_{\mathbb{Z}_n}^d$.  It is not
	hard to show that if a simple tensor covers at least two non-zero
	entries in $T_{\mathbb{Z}_n}^d$ then it places a positive weight on a
	zero-entry of $T_{\mathbb{Z}_n}^d$.  As this cannot be canceled out in a
	monotone computation, each simple tensor must cover at most one non-zero
	entry.  As there are $n^{d-1}$ such entries, the result follows. 
\end{proof}

\section{Acknowledgements}

We would like to thank Swastik Kopparty for alerting us to the standard construction
presented in Proposition~\ref{prop:indlayers} and Madhu Sudan for pointing us to the existence of
Lemma~\ref{lem:f2-sparse-irred}.  We would also like to thank Scott Aaronson, Arnab Bhattacharyya,
Andy Drucker, Kevin Hughes, Neeraj Kayal, Satya Lokam, Guy Moshkovitz, and Jakob Nordstrom for various constructive
conversations.
\newpage

\bibliographystyle{amsurl}
\bibliography{tensor-rank}

\newpage
\appendix
\section{Basic Facts about Tensors}

We now prove some relevant facts about tensors that are needed for the rest of
the paper.

\begin{lem}
	A $\otimes_{j=1}^d \mathbb{F}^{n_j}$-tensor
	is an $\prod_{j=1}^dn_j$ dimensional $\mathbb{F}$-vector space, with standard basis
	$\{\otimes_{j=1}^d \vec{e}_{i_j,j}\}_{i_j\in[n_j]}$
	where $\{e_{i_j,j}\}_{i_j\in[n_j]}$ is the standard basis for 
	$\mathbb{F}^{n_j}$.
	\label{lem:standardbasis}
\end{lem}
\begin{proof}
	Recall that the tensor product space $\otimes_{j=1}^d\mathbb{F}^{n_d}$
	is the set of functions from $\prod_{j=1}^d[n_j]$ to $\mathbb{F}$. As a
	$\mathbb{F}$-valued function space, it is thus an
	$\mathbb{F}$-vector space.  That it is $\prod_{j=1}^dn_h$ dimensional follows from the
	fact that this is the cardinality of the domain.

	To see that the basis is as claimed, note that the function
	$\otimes_{j=1}^d\vec{e}_{i_j,j}$ is equal to the tensor
	$T(i'_1,\ldots,i'_d)=\prod_{j=1}^d \lib i'_j=i_j\rib$.
	It is then not hard to see that
	these tensors are a basis for the tensor product space.
\end{proof}

\begin{lem}[Multilinearity of Tensor Product]
	\label{lem:tensormultilinearity}
	Suppose $j\in[d]$, $\vec{v}_j\in \mathbb{F}^{n_j}$, and $a,b\in\mathbb{F}$.  In the tensor product space
	$\otimes_{j=1}^d\mathbb{F}^{n_j}$, for any $j_0\in [d]$ and
	$\vec{w}\in\mathbb{F}^{n_{j_0}}$ the following identity holds:
	\begin{equation*}
		\vec{v}_1\otimes\cdots\otimes(a\vec{v}_{j_0}+b\vec{w})\otimes\cdots\otimes\vec{v_d}=
		a(\vec{v}_1\otimes\cdots\otimes\vec{v}_{j_0}\otimes\cdots\otimes\vec{v_d})+
		b(\vec{v}_1\otimes\cdots\otimes\vec{w}\otimes\cdots\otimes\vec{v_j})
	\end{equation*}
\end{lem}
\begin{proof}
	This follows directly from Definition~\ref{defn:simpletensor}.
\end{proof}

We now use these properties to establish a class of rank-preserving maps on tensors.

\begin{lem}
	For $j\in[d]$, consider linear maps $A_j:\mathbb{F}^{n_j}\rightarrow
	\mathbb{F}^{n'_j}$.
	
	\begin{enumerate}

		\item The $A_j$ induce a function on simple tensors
		$\otimes_{j=1}^d \vec{v}_d\mapsto
		\otimes_{j=1}^d A_j\vec{v}_j$ which uniquely extends to a
		linear map on the tensor product spaces which is denoted
		$\otimes_{j=1}^d A_j:\otimes_{j=1}^d
		\mathbb{F}^{n_j}\rightarrow\otimes_{j=1}^d \mathbb{F}^{n'_j}$.\label{lem:induce:define}
		\item If the $A_j$ are invertible, then so is $\otimes_{j=1}^d A_j$
		and its inverse is given by $\otimes_{j=1}^d A_j^{-1}$.\label{lem:induce:inverse}
		\item For $T:\prod_{j=1}^d[n_j]^d\to\mathbb{F}$, $\rank(T)\ge
		\rank\bigl((\otimes_{j=1}^d A_j)(T)\bigr)$, with
		equality if the $A_j$ are invertible.\label{lem:induce:rank}
	\end{enumerate}
	\label{lem:induce}
\end{lem}
\begin{proof}

	(\ref{lem:induce:define}): By Lemma~\ref{lem:standardbasis} the tensor product
	space $\otimes_{j=1}^d\mathbb{F}^{n_j}$ has a basis consisting
	entire of simple tensors.  Thus by standard linear algebra, the map
	$\otimes_{j=1}^d \vec{v}_j\mapsto \otimes_{j=1}^d A_j\vec{v}_j$ on this basis extends uniquely to a linear map
	$\otimes_{j=1}^d A_j$ on the entire tensor product space.  

	It must also be shown that the map $\otimes_{j=1}^d A_j$ induced from the
	basis elements is also
	compatible with the map $\otimes_{j=1}^d\vec{v}_j\mapsto
	\otimes_{j=1}^dA_j\vec{v}_j$ defined
	on the rest of the simple tensors.  This fact follows from the linearity
	of the $A_j$ and
	the multilinearity of the tensor product, Lemma~\ref{lem:tensormultilinearity}.
	That is, we first use that each $\vec{v}_j$ can be expressed in terms of the basis
	elements $\vec{v}_j=\sum^{n_j}_{i_j=1}c_{i_j,j}\vec{e}_{i_j,j}$ and then notice that
	by multilinearity of the tensor product we have
	\begin{align*}
		\bigotimes_{j=1}^d A_j\vec{v}_j
			=& \bigotimes_{j=1}^d A_j\left(\sum^{n_j}_{i_j=1} c_{i_j,j}\vec{e}_{i_j,j}\right)
			= \bigotimes_{j=1}^d \left(\sum^{n_j}_{i_j=1} c_{i_j,j}A_j\vec{e}_{i_j,j}\right)\\
			=& \sum^{n_1}_{i_1=1}\cdots\sum^{n_d}_{i_d=1} c_{i_1,1}\cdots c_{i_d,d}\bigotimes_{j=1}^d A_j\vec{e}_{i_j,j}\\
			=& \sum^{n_1}_{i_1=1}\cdots\sum^{n_d}_{i_d=1} c_{i_1,1}\cdots c_{i_d,d}\bigotimes_{j=1}^d A_j\left(\otimes_{j=1}^d\vec{e}_{i_j,j}\right)
	\end{align*}

	We observe similarly that $\otimes_{j=1}^d \vec{v}_j
	=\sum^{n_1}_{i_1=1}\cdots\sum^{n_d}_{i_d=1} c_{i_1,1}\cdots c_{i_d,d}(\otimes_{j=1}^d
	\vec{e}_{i_j,j})$.  As the unique linear map induced above \textit{defines}
	$\bigotimes_{j=1}^d A_j(\otimes_{j=1}^d \vec{v}_j)$ as
	$\sum^{n_1}_{i_1=1}\cdots\sum^{n_d}_{i_d=1} c_{i_1,1}\cdots c_{i_d, d}\bigotimes_{j=1}^d
	A_j(\otimes_{j=1}^d \vec{e}_{i_j,j})$, this shows that $\otimes_{j=1}^d
	A_j\vec{v}_j=\bigotimes_{j=1}^d A_j(\otimes_{j=1}^d\vec{v}_j)$,
	and so the two maps agree on the simple tensors.

	It should also be noted that this argument is independent of the basis chosen, as
	long as the basis is chosen among the simple tensors.  This fact follows from the
	fact that the induced map on the entire space agrees with the map only defined on
	the simple tensors.  Thus, the map $\bigotimes_{j=1}^d A_j$ is well-defined.

	(\ref{lem:induce:inverse}):  Denote the linear maps $A:= \otimes_{j=1}^d A_j$, and
	$A^{-1}:=\otimes_{j=1}^d A_j^{-1}$.  Part~\ref{lem:induce:define} of this lemma shows that the maps $A$ and
	$A^{-1}$ compose, in either order, to be the identity on the simple tensors. As
	there is a basis among the simple tensors, by Lemma~\ref{lem:standardbasis}, this
	means that $A\circ A^{-1}$ and $A^{-1}\circ A$ are both identity maps.  Thus
	$A^{-1}$ is indeed the inverse map of $A$.

	(\ref{lem:induce:rank}): Consider a minimal simple tensor decomposition of $T$, so that
	$T=\sum^r_{l=1} \otimes_{j=1}^d\vec{v}_{j,l}$.  By
	part~\ref{lem:induce:define} of this lemma, we have a simple tensor decomposition
	$(\otimes_{j=1}^d A_j)T=\sum^r_{l=1}
	\otimes_{j=1}^d A_j\vec{v}_{j,l}$.  This establishes the
	desired rank inequality.  To establish equality when the $A_j$ are invertible it is
	enough to run the inequality in the opposite direction using the linear map
	$\otimes_{j=1}^d A_j^{-1}$ and using part~\ref{lem:induce:inverse} of
	this lemma.
\end{proof}

We now use these rank-preserving maps to establish facts about tensors and their layers.

\begin{lem}
	\label{lem:simplelayers}
	Consider
	$T=\otimes_{j=1}^d\vec{v}_j\in\bigotimes_{j=1}^d\mathbb{F}^{n_j}$,
	where $T$ is split into layers as $T=[T_1|\cdots|T_{n_d}]$.  Then
	$T_l=(\vec{v}_1\otimes\cdots\otimes\vec{v}_{d-1})\cdot\vec{v}_d(l)
	\in\bigotimes_{j=1}^{d-1}\mathbb{F}^{n_j}$.
\end{lem}
\begin{proof}
	Definition~\ref{defn:layers} and Definition~\ref{defn:simpletensor} show that
	$T_l(i_1,\ldots,i_{d-1}):=T(i_1,\ldots,i_{d-1},l)=\vec{v}_1(i_1)\cdots\vec{v}_{d-1}(i_{d-1})\cdot\vec{v}_d(l)$.
	We can then note that this is exactly the
	function $\vec{v}_1\otimes\cdots\otimes\vec{v}_{d-1}$, multiplied by the scalar
	$\vec{v}_d(l)$, as desired.
\end{proof}

\begin{lem}
	\label{lem:layermap}
	Consider the operation of taking the $l$-th layer (along the $d$-th axis).  This is
	a linear map
	$L_l:\bigotimes_{j=1}^d\mathbb{F}^{n_j}\rightarrow\bigotimes_{j=1}^{d-1}\mathbb{F}^{n_j}$.
\end{lem}
\begin{proof}
	Given the tensor $T(\cdot,\ldots,\cdot)$.  Taking the $l$-th layer yields
	$T_l(\cdot,\ldots,\cdot):=T(\cdot,\ldots,\cdot,l)$.  Thus, the statements
	$T=S+R\implies T_l=S_l+R_l$, and $c\in\mathbb{F}, T=cS\implies T_l=cS_l$ hold
	because they are simply a restriction of the above identity.
\end{proof}

We can now prove the main lemma of this appendix, on how applying linear maps interacts with the
layers of a tensor.

\begin{lem}
	\label{lem:maplayers}
	Consider $T\in\otimes_{j=1}^d\mathbb{F}^{n_j}$.
	Expand $T$ into layers, so $T=[T_1|\cdots|T_{n_d}]$.
	
	Let $(a_{i,j})_{i,j}\in\mathbb{F}^{m\times n_d}$ be
	a matrix. Define $A:\mathbb{F}^{n_d}\rightarrow\mathbb{F}^{m}$ to be the
	linear map the matrix $(a_{i,j})_{i,j}$ induces via the standard basis. Then, 
	\begin{equation*}
		(I\otimes\cdots\otimes I\otimes	A)(T)
		=\left[\sum_{i_1=1}^{n_d}a_{1,i_1}T_{i_1}\left|\cdots\left|\sum_{i_{m}=1}^{n_d}a_{m,i_m}T_{i_m}\right.\right.\right]
	\end{equation*}
\end{lem}
\begin{proof}
	The proof is in two parts.  The first part proves the claim for simple tensors, and the
	second part extends the claim, using the linearity shown in Lemma~\ref{lem:layermap}, to
	general case.

	We first prove the claim for simple tensors. Let
	$T=\otimes_{j=1}^d\vec{v}_j$ be a simple tensor.	Let
	$\{\vec{e}_{i,d}\}_{i\in[n_d]}$ be the standard basis for $\mathbb{F}^{n_d}$ and
	$\{\vec{e}_{i',d}\}_{i'\in[m]}$ be the standard basis for $\mathbb{F}^{m}$.
	Then by expanding out in terms of the basis elements and using multilinearity, we
	have
	\begin{align*}
		T
		&=\vec{v}_{1}\otimes\cdots\otimes\vec{v}_d\\
		&=\sum_{i=1}^{n_d}\vec{v}_{1}\otimes\cdots\otimes\vec{v}_{d-1}\otimes(\vec{v}_d(i)\vec{e}_{i,d})
	\intertext{Denote $T':=(I\otimes\cdots\otimes I\otimes A)(T)$.  So then,}
		T'
		&=\sum_{i=1}^{n_d}\vec{v}_{1}\otimes\cdots\otimes\vec{v}_{d-1}\otimes A(\vec{v}_d(i)\vec{e}_{i,d})\\
		&=\sum_{i=1}^{n_d}\vec{v}_{1}\otimes\cdots\otimes\vec{v}_{d-1}\otimes (\vec{v}_d(i)\cdot A(\vec{e}_{i,d}))\\
		&=\sum_{i=1}^{n_d}\vec{v}_{1}\otimes\cdots\otimes\vec{v}_{d-1}\otimes \left(\vec{v}_d(i)\cdot \sum_{i'=1}^{m}a_{i',i}\vec{e}_{i',d}\right)\\
		&=\sum_{i=1}^{n_d}\sum_{i'=1}^{m}\vec{v}_d(i)\cdot a_{i',i}\cdot (\vec{v}_{1}\otimes\cdots\otimes\vec{v}_{d-1}\otimes \vec{e}_{i',d})
	\intertext{By Lemma~\ref{lem:simplelayers} and Lemma~\ref{lem:layermap}, we have,}
		T'_l
		&=\sum_{i=1}^{n_d}\sum_{i'=1}^{n_d'}\vec{v}_d(i)\cdot a_{i',i}\cdot (\vec{v}_{1}\otimes\cdots\otimes\vec{v}_{d-1})\cdot \vec{e}_{i',d}(l)
	\intertext{and using that $\vec{e}_{i',d}(l)=\lib i'=l\rib$,}
		&=\sum_{i=1}^{n_d}\vec{v}_d(i)\cdot a_{l,i}\cdot (\vec{v}_{1}\otimes\cdots\otimes\vec{v}_{d-1})\\
		&=\sum_{i=1}^{n_d}a_{l,i}T_i
	\end{align*}
	which establishes the claim for simple tensors.
	
	Now let $T\in\otimes_{j=1}^d\mathbb{F}^{n_j}$ be an arbitrary
	tensor.  Consider a simple tensor expansion $T=\sum_{k=1}^r S_k$ for
	$S_k=\otimes_{j=1}^d\vec{v}_{j,k}$.
	Denote $S_{k,l}$ to be the $l$-th layer of $S_k$.  So then as the $S_k$ are simple,
	we have that $(I\otimes\cdots\otimes I\otimes
	A)(S_k)=\left[\sum_{i_1=1}^{n_d}a_{1,i_1}S_{k,i_1}\left|\cdots\left|\sum_{i_m=1}^{n_d}a_{m,i_m}S_{k,i_m}\right.\right.\right]$
	by the above analysis.  So then,
	\begin{align*}
		(I\otimes\cdots\otimes I\otimes A)(T)
			&=(I\otimes\cdots\otimes I\otimes A)\left(\sum_{k=1}^r S_k\right)\\
			&=\sum_{k=1}^r \left[\sum_{i_1=1}^{n_d}a_{1,i_1}S_{k,i_1}\left|\cdots\left|\sum_{i_m=1}^{n_d}a_{m,i_m}S_{k,i_m}\right.\right.\right]\\
		\intertext{by linearity of taking layers, Lemma~\ref{lem:layermap}, we get}
			&=\left[\sum_{k=1}^r\sum_{i_1=1}^{n_d}a_{1,i_1}S_{k,i_1}\left|\cdots\left|\sum_{k=1}^r\sum_{i_m=1}^{n_d}a_{m,i_1}S_{k,i_m}\right.\right.\right]\\
			&=\left[\sum_{i_1=1}^{n_d}a_{1,i_1}\left(\sum_{k=1}^rS_{k,i_1}\right)\left|\cdots\left|\sum_{i_m=1}^{n_d}a_{m,i_m}\left(\sum_{k=1}^rS_{k,i_m}\right)\right.\right.\right]\\
			&=\left[\sum_{i_1=1}^{n_d}a_{1,i_1}T_{i_1}\left|\cdots\left|\sum_{i_m=1}^{n_d}a_{m,i_m}T_{i_m}\right.\right.\right]
	\end{align*}
	which is the desired result.
\end{proof}

We now apply this to get a symmetry lemma.

\begin{cor}
	\label{lem:reorder}
	Consider $T\in\otimes_{j=1}^d\mathbb{F}^{n_j}$.
	Expand $T$ into layers, so $T=[T_1|\cdots|T_{n_d}]$.  For any permutation
	$\sigma:[n_d]\rightarrow[n_d]$,
	\[\rank([T_{\sigma(1)}|\cdots|T_{\sigma(n_d)}])=\rank([T_1|\cdots|T_{n_d}])\]
\end{cor}
\begin{proof}
	Let $P$ be the linear transformation defined by the permutation that $\sigma$
	induces on the basis vectors of $\mathbb{F}^{n_d}$.  Then $P$ is invertible, and so
	by Lemma~\ref{lem:induce}.\ref{lem:induce:inverse} the induced transformation
	$I\otimes\cdots\otimes I\otimes P$ is also invertible and so
	$\rank((I\otimes\cdots\otimes I \otimes P)(T))=\rank(T)$ by
	Lemma~\ref{lem:induce}.\ref{lem:induce:inverse}.  Then, by Lemma~\ref{lem:layermap}
	we see that $(I\otimes\cdots\otimes I\otimes P)(T)=[T_1|\cdots|T_{n_d}]$.
\end{proof}

We also need another symmetry lemma.

\begin{lem}
	\label{lem:rotateaxes}
	For $T\in\otimes_{j=1}^d\mathbb{F}^{n_j}$ and a permutation
	$\sigma:[d]\rightarrow [d]$, define
	$T'\in\otimes_{j=1}^d\mathbb{F}^{n_{\sigma(j)}}$ by
	$T'(i_1,\ldots,i_d)=T(i_{\sigma^{-1}(1)},\ldots,i_{\sigma^{-1}(d)})$.  Then,
	$\rank(T)=\rank(T')$.
\end{lem}
\begin{proof}
	We show $\rank(T)\ge\rank(T')$, and the equality follows by symmetry as $\sigma$ is
	invertible.  Consider a simple tensor decomposition
	$T=\sum_{k=1}^r\otimes_{j=1}^d\vec{v}_{j,k}$. It is then easy to
	see that
	$T'=\sum_{k=1}^r\otimes\vec{v}_{\sigma(j),k}$ by
	considering the equation pointwise:
	$T'(i_1,\ldots,i_d)=T(i_{\sigma^{-1}(1)},\ldots,i_{\sigma^{-1}(d)})=\sum_{k=1}^r
	\prod_{j=1}^d\vec{v}_{j,k}(i_{\sigma^{-1}(j)})=\sum_{k=1}^r\prod_{j=1}^d\vec{v}_{\sigma(j),k}(i_j)$. The conclusion then follows by considering a
	minimal rank expansion.
\end{proof}

Finally, we need a corollary about how dropping layers from a tensor affects rank.

\begin{cor}
	\label{cor:droplayers}
	For layers $S_1,\ldots,S_{n_d},S'\in\bigotimes_{j=1}^{d-1}\mathbb{F}^{n_j}$, we have that 
	\begin{equation*}
		\rank([S_1|\cdots|S_{n_d}])\le\rank([S_1|\cdots|S_{n_d}|S'])
	\end{equation*}
	with equality if $S'$ is the zero layer.
\end{cor}
\begin{proof}
	($\le$): The projection map $P$ induces the map
	$(I\otimes\cdots\otimes I\otimes P)$ which takes
	$[S_1|\cdots|S_{n_d}|S']$ to $[S_1|\cdots|S_{n_d}]$ by Lemma~\ref{lem:layermap} and
	so Lemma~\ref{lem:induce}.\ref{lem:induce:rank} implies that the rank has not
	increased.  

	($\ge$): So now assume $S'$ is the zero layer.  Then again we apply
	Lemma's~\ref{lem:layermap} and Lemma~\ref{lem:induce}.\ref{lem:induce:rank} but now
	extend 
	the natural inclusion map $\iota:\mathbb{F}^{n_d}\rightarrow \mathbb{F}^{n_d+1}$ to
	a linear map $(I\otimes\cdots\otimes I\otimes \iota)$ on the tensors which takes
	$[S_1|\cdots|S_{n_d}]$ to $[S_1|\cdots|S_{n_d}|0]$, again showing that the rank has
	not increased.
\end{proof}

\section{Layer Reduction}\label{sect:layerreduction}

This section details a generalization of row-reduction, which we call layer-reduction.  We show that
layer-reduction can alter a tensor in such a way to provably reduce its rank.  By showing this
process can be repeated many times, a rank lower bound can be established.

The following lemma is the main technical part of this section.  H\aa stad
implicitly used\footnote{H{\aa}stad's usage, and proof, is reflected by Lemmas~2, 3 and 4 (and the
following discussion) of the conference version~\cite{20100105.1}.  The journal
version~\cite{20100105.2} ascribes the origin of these lemmas to Lemma~2 in the work of Hopcroft and
Kerr~\cite{hopcroft-kerr}} a version of this lemma in his proof that tensor rank is
\NP-Complete~\cite{20100105.1,20100105.2}	However, H{\aa}stad's usage requires that $S_{n_d}$
is a rank-one tensor.  This special case does not seem to directly imply our lemma, which
was independently proven.  While the special case is sufficient to lower-bound the
combinatorially-constructed tensors of Section~\ref{sect:combtensors}, the full lemma
is needed to lower-bound the rank of the algebraically-constructed tensors of
Section~\ref{sect:algtensors}.

\begin{lem}[Layer Reduction]
	\label{lem:layerreduction}
	For layers $S_1,\ldots,S_{n_d}\in\bigotimes_{j=1}^{d-1}\mathbb{F}^{n_j}$
	with $S_{n_d}$ non-zero, there exist constants $c_1,\ldots,c_{n_d-1}\in\mathbb{F}$ such that
	\begin{equation*}
		\rank([S_1|\cdots | S_{n_d}])\ge \rank([S_1+c_1S_{n_d}|\cdots | S_{n_d-1}+c_{n_d-1}S_{n_d}])+1
	\end{equation*}
\end{lem}
\begin{proof}
	Denote $T:=[S_1|\ldots | S_{n_d}]$.  The proof is in two steps.  The first step
	defines a linear transformation $A$ on $\mathbb{F}^{n_d}$ such that the linear
	transformation $I\otimes\cdots \otimes I\otimes A$ is a higher-dimensional analogue
	of a row-reduction step in Gaussian elimination.  That is, for $T'$ the image of
	$T$, it is seen that $T'=[S_1+c_1S_{n_d}|\ldots |
	S_{n_d-1}+c_{n_d-1}S_{n_d}|S_{n_d}]$ by Lemma~\ref{lem:maplayers}.  The $c_i$ are chosen in such a way so that
	$T'$ has a minimal simple tensor expansion where some simple tensor $R$ is non-zero
	only on the $S_{n_d}$-layer.  In the second step, the $S_{n_d}$-layer is dropped and the remaining tensor
	$T''=[S_1+c_1S_{n_d}|\ldots | S_{n_d-1}+c_{n_d-1}S_{n_d}]$ no longer requires $R$ in
	its simple tensor expansion and so $\rank(T)\ge \rank(T'')+1$.
	
	Consider a minimal simple tensor expansion
	$T=\sum_{k=1}^r \otimes_{j=1}^d\vec{v}_{j,k}$.  Expanding the
	$\vec{v}_{d,k}$ in terms of basis vectors yields \[T=\sum_{k=1}^r
	(\vec{v}_{1,k}\otimes\cdots\otimes\vec{v}_{d-1,k})\otimes
	(\vec{v}_{d,k}(1)\cdot\vec{e}_{1,d}+\cdots+\vec{v}_{d,k}(n_d)\cdot\vec{e}_{n_d,d}))\]
	and in particular Lemma~\ref{lem:simplelayers} shows that $S_{n_d}=\sum_{k=1}^r
	(\vec{v}_{1,k}\otimes\cdots\otimes\vec{v}_{d-1,k})\cdot\vec{v}_{d,k}(n_d)$. As
	$S_{n_d}$ is
	non-zero there must be some $k_0$ such that $\vec{v}_{d,k_0}(n_d)\ne 0$.  Define
	$A:\mathbb{F}^{n_d}\rightarrow \mathbb{F}^{n_d}$ to be the linear transformation
	defined by its action on the standard basis
	\begin{equation*}
		A(\vec{e}_{i,d}) =
			\begin{cases}
				\vec{e}_{n_d,d}-\frac{\vec{v}_{d,k_0}(1)}{\vec{v}_{d,k_0}(n_d)}\vec{e}_{1,d}-\cdots-\frac{\vec{v}_{d,k_0}(n_d-1)}{\vec{v}_{d,k_0}(n_d)}\vec{e}_{n_d-1,d}	&	\text{if } i=n_d\\
				\vec{e}_{i,d}	&	\text{else}
			\end{cases}
	\end{equation*}
	
	Letting $I$ denote the identity transformation, consider the tensor
	$T':=(I\otimes\cdots\otimes I\otimes A)(T)\in
	\mathbb{F}^{n_1}\otimes\cdots\otimes\mathbb{F}^{n_d}$. By Lemma~\ref{lem:maplayers},
	we observe that $T'=[S_1+c_1S_{n_d}|\cdots | S_{n_d-1}+c_{n_d-1}S_{n_d}|S_{n_d}]$,
	where $c_i=-\frac{\vec{v}_{d,k_0}(j)}{\vec{v}_{d,k_0}(n_d)}$.

	By Lemma~\ref{lem:induce} we have the simple tensor expansion $T'=\sum_{k=1}^r
	\vec{v}_{1,k}\otimes\cdots\otimes \vec{v}_{d-1,k}\otimes A\vec{v}_{d,k}$.  By
	construction, $A(\vec{v}_{d,k_0})=\vec{v}_{d,k_0}(n_d)\cdot\vec{e}_{n_d,d}$.  Using
	Lemma~\ref{lem:simplelayers} we observe that the simple tensor
	$\vec{v}_{1,k_0}\otimes\cdots\otimes \vec{v}_{d-1,k_0}\otimes A\vec{v}_{d,k_0}$ has
	non-zero entries only on the $S_{n_d}$-layer.

	We now define the linear transformation $A':\mathbb{F}^{n_d}\rightarrow
	\mathbb{F}^{n_d-1}$ defined by
	\begin{equation*}
		A'(\vec{e}_{i,d}) =
			\begin{cases}
				\vec{0}		&	\text{if } i=n_d\\
				\vec{e}_{i,d}	&	\text{else}
			\end{cases}
	\end{equation*}
	This will correspond to dropping the $S_{n_d}$-layer.  We can compose this with $A$
	to get $A''=A'\circ A$, defined by 
	\begin{equation*}
		A''(\vec{e}_{i,d}) =
			\begin{cases}
				-\frac{\vec{v}_{d,k_0}(1)}{\vec{v}_{d,k_0}(n_d)}\vec{e}_{1,d}-\cdots-\frac{\vec{v}_{d,k_0}(n_d-1)}{\vec{v}_{d,k_0}(n_d)}\vec{e}_{n_d-1,d}	&	\text{if } i=n_d\\
				\vec{e}_{i,d}	&	\text{else}
			\end{cases}
	\end{equation*}
	So now we take $T''=(I\otimes\cdots\otimes I\otimes A'')(T)$.  By
	Lemma~\ref{lem:maplayers} we see that $T''=[S_1+c_1S_{n_d}|\cdots |
	S_{n_d-1}+c_{n_d-1}S_{n_d}]$.  Further, we observe now that by construction
	$A''(\vec{v}_{d,k_0})=\vec{0}$.  This leads to the simple tensor expansion,
	\begin{align*}
		T''
			&= \sum_{k=1}^r \vec{v}_{1,k}\otimes\cdots\otimes \vec{v}_{d-1,k}\otimes A''\vec{v}_{d,k}\\
			&= \vec{v}_{1,k_0}\otimes\cdots\otimes \vec{v}_{d-1,k_0}\otimes A''\vec{v}_{d,k_0}+\sum_{k=1,k\ne k_0}^r \vec{v}_{1,k}\otimes\cdots\otimes \vec{v}_{d-1,k}\otimes A''\vec{v}_{d,k}\\
			&= \vec{v}_{1,k_0}\otimes\cdots\otimes \vec{v}_{d-1,k_0}\otimes \vec{0}+\sum_{k=1,k\ne k_0}^r \vec{v}_{1,k}\otimes\cdots\otimes \vec{v}_{d-1,k}\otimes A''\vec{v}_{d,k}\\
			&= \sum_{k=1,k\ne k_0}^r \vec{v}_{1,k}\otimes\cdots\otimes \vec{v}_{d-1,k}\otimes A''\vec{v}_{d,k}
	\end{align*}
	Therefore $\rank(T'')\le r-1=\rank(T)-1$, and thus $\rank(T)\ge \rank(T'')-1$.
\end{proof}

The layer-reduction lemma will mostly be used via the following extension.

\begin{cor}[Iterative Layer-Reduction]
	\label{cor:layerreduction}
	For layers $S_1,\ldots,S_{n_d}\in
	\mathbb{F}^{n_1}\otimes\cdots\otimes\mathbb{F}^{n_{d-1}}$ with $S_1,\ldots,S_m$
	linearly independent (as vectors in the space $\mathbb{F}^{n_1\cdots n_{d-1}}$),
	there exist constants $c_{i,j}\in\mathbb{F}$, $i\in\{1,\ldots,m\}$,
	$j\in\{m+1,\ldots, n_d\}$, such that
	\begin{equation}
		\label{eq:maineq}
		\rank([S_1|\ldots | S_{n_d}])\ge
		\rank\left(\left[S_{m+1}+\sum_{i=1}^mc_{i,m+1}S_i\left|\ldots \left|
		S_{n_d}+\sum_{i=1}^mc_{i,n_d}S_i\right.\right.\right]\right)+m
	\end{equation}
\end{cor}
\begin{proof}
	The proof is by induction on $m$.

	\underline{$m=1$:} This is Lemma~\ref{lem:layerreduction}, up to reordering of the
	layers, with the observation that the singleton set $\{S_1\}$ is
	linearly-independent iff $S_1$ is non-zero. The reordering of layers is justified by
	Lemma~\ref{lem:reorder}.

	\underline{$m>1$:} By the induction hypothesis we have that 
	\begin{equation}
		\label{eq:inductionhyp}
		\rank([S_1|\ldots | S_{n_d}])\ge
		\rank\left(\left[S_{m}+\sum_{i=1}^{m-1}c_{i,m}S_i\left|\ldots \left|
		S_{n_d}+\sum_{i=1}^{m-1}c_{i,n_d}S_i\right.\right.\right]\right)+m-1
	\end{equation}
	for the appropriate set of constants $c_{i,j}$.
	As the $S_i$ are linearly independent, $S_{m}+\sum_{i=1}^{m-1}c_{i,m}S_i$ is
	non-zero and so we can eliminate this layer from
	$\left[S_{m}+\sum_{i=1}^{m-1}c_{i,m}S_i\left|\ldots \left|
	S_{n_d}+\sum_{i=1}^{m-1}c_{i,n_d}S_i\right.\right.\right]$ by
	Lemma~\ref{lem:layerreduction} and consequently have
	\begin{multline}
		\label{eq:layerinduction}
		\rank\left(\left[S_{m}+\sum_{i=1}^{m-1}c_{i,m}S_i\left|\ldots \left|
		S_{n_d}+\sum_{i=1}^{m-1}c_{i,n_d}S_i\right.\right.\right]\right)\\
		\ge
		\rank\left(\left[\left(S_{m+1}+\sum_{i=1}^{m-1}c_{i,m+1}S_i\right)+c_{m,m+1}\left(S_{m}+\sum_{i=1}^{m-1}c_{i,m}S_i\right)\right|\right.\\
		\ldots \left.\left| \left(S_{n_d}+\sum_{i=1}^{m-1}c_{i,n_d}S_i\right)+c_{m,n_d}\left(S_{m}+\sum_{i=1}^{m-1}c_{i,m}S_i\right)\right]\right)+1
	\end{multline}
	where the $c_{m,j}$ are new constants.  Now define
	\begin{equation}
		\label{eq:newconstants}
		c'_{i,j}=
			\begin{cases}
				c_{i,j}+c_{m,j}c_{i,m}	&	\text{if } i\ne m\\
				c_{m,j}			&	\text{else}
			\end{cases}
	\end{equation}
	Combining Equations~\eqref{eq:inductionhyp}, \eqref{eq:layerinduction}, and
	\eqref{eq:newconstants} yields the desired Equation~\eqref{eq:maineq}.
\end{proof}

Notice that by Lemma~\ref{lem:rotateaxes} we can in fact use Lemma~\ref{lem:layerreduction}
and Corollary~\ref{cor:layerreduction} along any axis, not just the $d$-th one.

\begin{rmk}
	\label{rmk:layerreductionbarrier}
	Lemma~\ref{lem:layerreduction} shows that the rank of
	$T:\prod_{j=1}^d[n_j]\rightarrow \mathbb{F}$ is at least $1$ more than the rank of
	some
	$T':\prod_{j=1}^d[n_j']\rightarrow \mathbb{F}$, where $n_j'=n_j$ for all $j\ne j_0$,
	and $n_{j_0}'=n_{j_0}-1$.  In using this lemma, the quantity $\sum_{j=1}^d n_j$
	decreases by one.  Therefore, we can never hope to apply this lemma more than
	$\sum_{j=1}^d n_j$ many times, and thus using this lemma alone will never produce
	lower bounds larger than this quantity.
	Corollary~\ref{cor:layerreduction} simply applies Lemma~\ref{lem:layerreduction}, so
	the same barriers apply.
\end{rmk}

\section{Proofs for Section~\ref{sect:combtensors}}\label{sect:combtensorsproofs}

\begin{proof}[Proof of Lemma~\ref{lem:addident}]
	Notice that the left hand sides of Equation~\ref{eq:addidentrankeven} and
	Equation~\ref{eq:addidentrankodd} are equal. This follows from applying
	Corollary~\ref{cor:droplayers} twice (using that this corollary extends to layers
	along any axes,
	not just the $d$-th, by applying Lemma~\ref{lem:rotateaxes}), once on the layers
	slicing the page vertically, and once on the layers slicing the page horizontally.  Thus, it is enough to show
	Equation~\ref{eq:addidentrankeven}.

	We now apply Corollary~\ref{cor:layerreduction}.  First, we use it on the layers
	slicing the page vertically and deriving that
	\begin{equation}
		\label{eq:ident1}
		\rank\left(\left[
			\begin{matrix}
				I_n & 0_n\\
				0_n & I_n
			\end{matrix}
			\left|
			\begin{matrix}
				0_n & 0_n\\
				A_1 & 0_n
			\end{matrix}
			\right|
			\cdots
			\left|
			\begin{matrix}
				0_n & 0_n\\
				A_k & 0_n
			\end{matrix}
			\right.
			\right]
			\right)
		\ge
		\rank\left(\left[
			\begin{matrix}
				I_n\\
				C
			\end{matrix}
			\left|
			\begin{matrix}
				0_n\\
				A_1
			\end{matrix}
			\right|
			\cdots
			\left|
			\begin{matrix}
				0_n\\
				A_k
			\end{matrix}
			\right.
			\right]
			\right)
			+n
	\end{equation}
	where $C$ is an $n\times n$ matrix of field elements defined by the constants
	$c_{i,j}$ of Corollary~\ref{cor:layerreduction}. Notice that the layers being
	dropped in the use of this corollary must be linearly independent.  However, as they
	are the layers of $[I_n|0_n|\cdots|0_n]$ which slice the page vertically, they have
	exactly one 1 in the first row\footnote{It is immaterial whether we call this a
	``row'' or ``column'', as no specific orientation of these tensors was chosen.}, and
	have 0 entries elsewhere.  As their non-zero entries are in different positions,
	they are linearly independent.  Similarly, we can apply the corollary again on the
	remaining layers that slice the page horizontally to see that
	\begin{equation}
		\label{eq:ident2}
		\rank\left(\left[
			\begin{matrix}
				I_n\\
				C
			\end{matrix}
			\left|
			\begin{matrix}
				0_n\\
				A_1
			\end{matrix}
			\right|
			\cdots
			\left|
			\begin{matrix}
				0_n\\
				A_k
			\end{matrix}
			\right.
			\right]
			\right)
		\ge
		\rank([C+C'|A_1|\cdots|A_k])
			+n
	\end{equation}
	where $C'$ is yet another $n\times n$ matrix of field elements produced by
	Corollary~\ref{cor:layerreduction}. We now invoke Corollary~\ref{cor:droplayers} to
	observe that
	\begin{equation}
		\label{eq:ident3}
		\rank([C+C'|A_1|\cdots|A_k])
		\ge
		\rank([A_1|\cdots|A_k])
	\end{equation}
	Combining Equations~\eqref{eq:ident1}, \eqref{eq:ident2}, and \eqref{eq:ident3}
	yields Equation~\eqref{eq:addidentrankeven} and thus the claim.
\end{proof}

\begin{proof}[Proof of Theorem~\ref{thm:mainthm}]
	\underline{(\ref{mainthm:size})}: This is clear from construction.

	\underline{(\ref{mainthm:rank})}:  We first note that $\lfloor \lg 2n\rfloor=\lfloor
	\lg n+1\rfloor=\lfloor \lg n\rfloor +1$.  We first prove the upper bound, and then the
	lower bound.

	To see that $\rank(T_n)\le 2n-2H(n)+1$ we observe that $T_n$ has exactly this many
	non-zero entries.  Denote this quantity $r_n$.  We proceed by induction on the
	recursive definition of the $S_{n,i}$.  For $n=1$,
	there is clearly exactly $2\cdot 1-2H(1)+1=1$ non-zero entry.  For $2n>1$, $r_{2n}=r_n+2n$
	which by induction yields $r_{2n}=(2n-2H(n)+1)+2n$.  Observing that $H(n)=H(2n)$, we see
	that $r_{2n}=2(2n)-2H(2n)+1$.  For $2n+1>1$, $r_{2n+1}=r_n+2n$, which by induction yields
	$r_{2n+1}=(2n-2H(n)+1)+2n$.  Noticing that $H(2n+1)=H(n)+1$ we have that
	$r_{2n+1}=2n-2(H(2n+1)-1)+1+2n=4n+2-2H(2n+1)+1=2(2n+1)-2H(2n+1)+1$. Thus, the
	induction hypothesis shows that $r_n=2n-2H(n)+1$ for all $n$, and thus
	upper-bounding the rank by this quantity.
	
	For the rank lower bound, we use Lemma~\ref{lem:addident} and induction on the
	recursive definition of the $S_{n,i}$.  Clearly,
	$\rank(T_1)\ge 1$.  Then for $2n>1$, $\rank(T_{2n})\ge\rank(T_n)+2n$, and for
	$2n+1>1$, $\rank(T_{2n+1})\ge\rank(T_n)+2n$.  These are exactly the same recurrences
	from the proceeding paragraph, and so they have the same solution: $\rank(T_n)\ge
	2n-2H(n)+1$.

	Combining these two bounds shows that $\rank(T_n)=2n-2H(n)+1$.

	\underline{(\ref{mainthm:explicit})}: This is clear from the equations defining the
	$S_{n,i}$.
\end{proof}

\begin{proof}[Proof of Corollary~\ref{cor:maincor}]
	\underline{(\ref{maincor:size})}: This is clear from construction.

	\underline{(\ref{maincor:rank})}: Observe that in the construction of $T'_n$, the
	matrices $S'_{n,i}$ for $i> \lfloor \lg (n-1)\rfloor +1$ are linearly independent.
	Thus, applying Corollary~\ref{cor:layerreduction}, we see that
	\begin{equation*}
		\rank(T'_n)\ge
		\rank([\tilde{S}'_{n,1}|\cdots|\tilde{S}'_{n,\lfloor\lg(n-1)\rfloor+1}])+n-(\lfloor\lg(n-1)\rfloor+1)
	\end{equation*}
	where
	\begin{equation*}
		\tilde{S}'_{n,i}=
			\begin{bmatrix}
				S_{n-1,i-1} & \vec{c}_i\\
				0	& 0
			\end{bmatrix}
	\end{equation*}
	for some arbitrary vectors $\vec{c}_i\in\mathbb{F}^{n-1}$.  It follows from
	Corollary~\ref{cor:droplayers} that we can drop the bottom row and last column of
	each of the $\tilde{S}_{n,i}$ without increasing the rank, so that 
	\begin{equation*}
		\rank([\tilde{S}'_{n,1}|\cdots|\tilde{S}'_{n,\lfloor\lg(n-1)\rfloor+1}])
			\ge
			\rank([S_{n-1,0}|\cdots|S_{n-1,\lfloor\lg(n-1)\rfloor}])
	\end{equation*}
	where the $S_{n-1,i-1}$ are as defined in Theorem~\ref{thm:mainthm}, and as such,
	$\rank([S_{n-1,0}|\cdots|S_{n-1,\lfloor\lg(n-1)\rfloor}])=2(n-1)+2H(n-1)+1$.
	Combining these inequalities yields the rank lower bound for $T'_n$.

	\underline{(\ref{maincor:explicit})}: This is clear from the equations defining
	$T'_n$, and using the explicitness of the $S_{n-1,i-1}$ as seen from
	Theorem~\ref{thm:mainthm}.
\end{proof}

\section{Algebraically-defined Tensors}\label{sect:algtensors}

The results of this section will be field-specific, and so we no longer work over an
arbitrary field.

\begin{lem}
	\label{lem:tensorelim}
	Let $\mathbb{F}_q$ be the field of $q$ elements. Consider $n\times n$ matrices
	$M_1,\ldots,M_k$ over $\mathbb{F}_q$ such that all non-zero linear combinations have
	full-rank. Then the tensor
	$T=[M_1|\cdots|M_k]$ has tensor rank at least
	$\frac{q^k-1}{q^k-q^{k-1}}n$.
\end{lem}
\begin{proof}
	The proof is via the probabilistic method, using randomness to perform an analogue
	of gate elimination.  For
	non-zero $\vec{c}\in\mathbb{F}_q^k$, the summation $\vec{c}\cdot\vec{M}$ will
	nullify terms in a simple tensor expansion with some probability.  This will in
	expectation reduce the rank.  We then invoke the hypothesis that the result is
	full-rank, to conclude the bound on the original rank.

	Consider a minimal simple tensor decomposition $T=\sum_{i=1}^r \vec{u}_i\otimes
	\vec{v}_i\otimes\vec{w}_i$.  For $\vec{c}\in\mathbb{F}_q^k$, consider
	(notation-abused) dot-product $\langle\vec{c},\vec{M}\rangle$, which can also be
	written as the matrix $\sum_{i=1}^k c_i M_i$.  By Lemma~\ref{lem:maplayers} it can
	be seen that this is the image of $T$ under the linear transformation $I\otimes
	I\otimes A$, where $A$ is the linear transformation that sends the basis element
	$\vec{e}_i$ to $c_i \vec{e}_1$.  Consequently, we have that $\langle\vec{c},
	\vec{M}\rangle=(I\otimes I \otimes A)(T)=\sum_{i=1}^r
	\vec{u}_i\otimes\vec{v}_i\otimes (A\vec{w}_i)$.  Noticing that
	$A\vec{w}_i=\langle\vec{c},\vec{w}_i\rangle$, and that we can then treat this as a
	matrix instead of one-layer tensor, we see that
	$\langle\vec{c},\vec{M}\rangle=\sum_{i=1}^r
	\langle\vec{c},\vec{w}_i\rangle\vec{u}_i\otimes\vec{v}_i$.
	
	Minimality implies that $\vec{u}_i\ne 0$ for all $i$. So for a fixed $i$, the set of
	$\vec{c}$ such that $\langle\vec{c},\vec{w}_i\rangle=0$ is a 1-dimensional subspace
	by the Rank-Nullity theorem.  Using that the field size is $q$, this shows that
	\begin{equation*}
		\Pr_{\vec{c}\in_u\mathbb{F}_q^k\setminus\{\vec{0}\}} 
		[\langle\vec{c},\vec{w}_i\rangle\ne 0]= \frac{q^k-q^{k-1}}{q^k-1}
	\end{equation*}
	Now define $S_{\vec{c}}:=\{i|\langle\vec{c},\vec{w}_i\rangle\ne 0\}$.    By
	linearity of expectation,
	$\mathbb{E}_{\vec{c}\in_u\mathbb{F}_q^k\setminus\{\vec{0}\}}
	[|S_{\vec{c}}|]=\frac{q^k-q^{k-1}}{q^k-1}r$.  Thus, there exits a non-zero $\vec{c}_0$
	such that $|S_{\vec{c}_0}|\le \frac{q^k-q^{k-1}}{q^k-1}r$.  Therefore, we can
	write the matrix $\langle\vec{c},\vec{M}\rangle$ as 
	$\langle\vec{c},\vec{M}\rangle=\sum_{i=1}^r
	\langle\vec{c},\vec{w}_i\rangle\vec{u}_i\otimes\vec{v}_i=\sum_{i\in S}
	\langle\vec{c},\vec{w}_i\rangle\vec{u}_i\otimes\vec{v}_i$.  The hypothesis on the
	$M_i$ says that $\langle \vec{c},\vec{M}\rangle$ is of full-rank, and therefore we
	have that $n\le\rank(\langle\vec{c},\vec{M}\rangle)\le |S|\le
	\frac{q^k-q^{k-1}}{q^k-1}r$.  As $\rank(T)=r$, we have that $\rank(T)\ge
	\frac{q^k-1}{q^k-q^{k-1}}n$.
\end{proof}

\begin{cor}
	\label{cor:alglb}
	Let $\mathbb{F}_q$ be the field of $q$ elements. Consider $n\times n$ matrices
	$M_1,\ldots,M_n$ over $\mathbb{F}_q$ such that all non-zero linear combinations have
	full-rank. Then the tensor
	$T=[M_1|\cdots|M_n]$ has tensor rank at least  $\frac{2q-1}{q-1}n-\lceil\log_q
	n\rceil-\frac{q}{q-1}=\frac{2q-1}{q-1}n-\Theta(\log_q n)$.
\end{cor}
\begin{proof}
	Let $k\le n$ be a parameter, to be optimized over later.

	Notice that the hypothesis show that the matrices $M_i$ are linearly independent and
	so Corollary~\ref{cor:layerreduction} shows that 
	\begin{equation*}
		\rank([M_1|\cdots|M_n])\ge\rank([M_1+M'_1|\cdots|M_k+M'_k])+(n-k)
	\end{equation*}
	where the $M'_i$ are linear combinations of the $M_{k+1},\ldots,M_n$.  Thus, any
	non-zero linear combination of the $(M_i+M'_i)$ is necessarily a non-zero linear
	combination of the $M_i$.  In particular, this shows that any non-zero linear
	combination of the $(M_i+M'_i)$ has full-rank. Thus, by Lemma~\ref{lem:tensorelim}
	\begin{equation*}
		\rank([M_1+M'_1|\cdots|M_k+M'_k])\ge \frac{q^k-1}{q^k-q^{k-1}}n
	\end{equation*}
	and so
	\begin{equation*}
		\rank([M_1|\cdots|M_n])\ge n-k+\frac{q^k-1}{q^k-q^{k-1}}n =: f(k)
	\end{equation*}

	One can observe that $k=\log_q\left(\frac{n\ln q}{1-1/q}\right)$ maximizes $f$, but
	asymptotically it is sufficient to take $k=\lceil \log_q n\rceil$.  Then
	\begin{align*}
		f(k)
			&=n-k+\frac{1-q^{-k}}{1-1/q}n\\
			&\ge n-\lceil\log_q n\rceil+\frac{1-1/n}{1-1/q}n\\
			&\ge n-\lceil\log_q n\rceil+\frac{1}{1-1/q}n-\frac{1/n}{1-1/q}n\\
			&\ge n-\lceil\log_q n\rceil+\frac{q}{q-1}n-\frac{q}{q-1}\\
			&\ge \frac{2q-1}{q-1}n-\lceil\log_q n\rceil-\frac{q}{q-1}
	\end{align*}
	As $f(k)$ lower-bounds the rank by the above, this establishes the claim.
\end{proof}

The above lemma and its corollary establish a property implying tensor rank lower bounds.
We now turn to constructing tensors that have this property.  Clearly we seek explicit
tensors, and by this we mean that each entry of the tensor is efficiently computable.

We first observe that the property can be easily constructed given explicit field extensions
of the base field $\mathbb{F}$.

\begin{prop}
	Let $\mathbb{F}$ be a field and $f\in\mathbb{F}[x]$ be an irreducible
	polynomial of degree $n$.  Then there exists $n\times n$ $\mathbb{F}$-matrices
	$M_1,\ldots,M_n$, such that all non-zero $\mathbb{F}$-linear combinations of the $M_i$ have
	full-rank. Furthermore, the entries of each matrix are computable in
	algebraic circuits of size $O(\polylog(n)\poly(\|f\|_0))$, where $\|f\|_0$ is the
	number of non-zero coefficients of $f$.
	\label{prop:indlayers}
\end{prop}
\begin{proof}
	Let $f(x)=a_nx^n+\cdots+a_1x+a_0$.  Recall that $\mathbb{K}=\mathbb{F}[x]/(f)$ is a
	field, and because $\deg f=n$, $\mathbb{K}$ is a $n$-dimensional
	$\mathbb{F}$-vector space, where we choose $1,x,\ldots,x^{n-1}$ as the
	basis. This gives an $\mathbb{F}$-algebra
	isomorphism $\mu$ between $\mathbb{K}$ and a sub-ring $M$ of the $m\times m$
	$\mathbb{F}$-matrices, where $M$ is defined as the image of $\mu$.  The
	map $\mu$ is defined by associating
	$\alpha\in\mathbb{K}$ with the matrix inducing the linear map
	$\mu(\alpha):\mathbb{F}^m\to\mathbb{F}^m$, where $\mu(\alpha)$ is the multiplication
	map of $\alpha$.  That is, using that $\mathbb{K}=\mathbb{F}^m$ we can see
	that the map $\beta\mapsto \alpha\beta$ for $\beta\in\mathbb{K}$ is an
	$\mathbb{F}$-linear map, and thus defines $\mu(\beta)$ over $\mathbb{F}^m$.

	That the map is injective follows from the fact that $\mu(\alpha)$ must
	map $1\in\mathbb{K}$ to $\alpha\in\mathbb{K}$, so $\alpha$ is
	recoverable from $\mu(\alpha)$ (and surjectivity follows be definition
	of $M$).  To see the required homomorphism properties is also not
	difficult.  As $(\alpha+\gamma)\beta=\alpha\beta+\gamma\beta$ for any
	$\alpha,\beta,\gamma\in\mathbb{K}$, this shows that
	$\mu(\alpha+\gamma)=\mu(\alpha)+\mu(\gamma)$ as linear maps, and thus as
	matrices.  Similarly, as $(\alpha\gamma)\beta=\alpha(\gamma\beta)$ for
	any $\alpha,\beta,\gamma\in\mathbb{K}$ it must be that
	$\mu(\alpha\gamma)=\mu(\alpha)\mu(\gamma)$.  That this map interacts
	linearly in $\mathbb{F}$ implies that it is an $\mathbb{F}$-algebra
	homomorphism, as desired.

	In particular, this means that $\alpha\in\mathbb{K}$ is invertible iff
	the matrix $\mu(\alpha)\in M\subseteq\mathbb{F}^{n\times n}$ is
	invertible.  As $\mathbb{K}$ is a field, the only non-invertible matrix
	in $M$ is $\mu(0)$.  The $\mathbb{F}$-algebra homomorphism means that
	for $a_i\in\mathbb{F}$ and $\alpha_i\in\mathbb{K}$, the linear
	combination $\sum a_i\mu(\alpha_i)$ equals $\mu(\sum a_i \alpha_i)$ and
	so the matrix $\sum a_i\mu(\alpha_i)$ is invertible iff $\sum
	a_i \alpha_i \ne 0$.  Thus, as $1,x,\ldots,x^{n-1}$ are $\mathbb{F}$-linearly
	independent in $\mathbb{K}$, it follows that the matrix
	$\mu(1),\mu(x),\ldots,\mu(x^{n-1})$ have that all non-zero $\mathbb{F}$-linear
	combinations are invertible, as desired.

	We now study how to compute $\mu(x^i)$.  Observe that acting as a linear
	map on $\mathbb{F}^n$, $\mu(x^i)$ sends $x^{i+j} \pmod{f}$.  To read off
	the $x^k$ component can be done with a lookup table to the coefficients
	of $f$, and thus in $O(\poly(\|f\|_0))$ size circuits.
\end{proof}

To make the above construction explicit, we need to show that the irreducible polynomial $f$
can be found efficiently.  We now cite the following result of Shoup~\cite{shoup}.  It says that we can find
irreducible polynomials in finite fields in polynomial-time provided that the field size is
fixed.

\begin{thm}[\cite{shoup}, Theorem 4.1]
	\label{thm:findirredpoly}
	For any prime or prime power $q$, an irreducible polynomial of degree $n$ in $\mathbb{F}_q[x]$ can be
	found in time $O(\poly(nq))$.
\end{thm}

Combining Corollary~\ref{cor:alglb}, Proposition~\ref{prop:indlayers}, and
Theorem~\ref{thm:findirredpoly}, we arrive at the following result.

\begin{cor}
	For any \textit{fixed} prime or prime power $q$, over the field $\mathbb{F}_q$ there is a family of tensors $T_n$ of size $[n]^3$ such that
	\begin{enumerate}
		\item $\rank(T_n)\ge \frac{2q-1}{q-1}n-\Theta(\log_q n)$
		\item On inputs $n$ and $(i,j,k)\in[n]^3$, $T_n(i,j,k)$ is computable in $O(\poly(nq))$
	\end{enumerate}
\end{cor}

Note that this is strictly worse than Corollary~\ref{cor:maincor} in two respects.  First,
while this result asymptotically matches the lower bound of Corollary~\ref{cor:maincor} over
$\mathbb{F}_2$, the above result is only valid over finite fields, and as the field size grows, the lower bound
approaches $2n-o(n)$.  This seems inherent in the approach. 

Further, the given construction is less explicit as computing even a single entry of the
tensor might require examining all $n$ of the coefficients in the irreducible polynomial
$f$, preventing a $O(\polylog(n))$ runtime.  One method of circumventing this problem is to
use \textit{sparse irreducible polynomials}.  In particular, we use the following well-known
construction.

\begin{lem}[\cite{van-lint}, Theorem 1.1.28]
	\label{lem:f2irredpoly}
	Over $\mathbb{F}_2[x]$, the polynomial
	\begin{equation*}
		f(x)=x^{2\cdot 3^l}+x^{3^l}+1
	\end{equation*}
	is irreducible for any $l\ge 0$.
	\label{lem:f2-sparse-irred}
\end{lem}

Observe that this allows for much faster arithmetic in the extension field, and this leads
to the following result when applying the above results.

\begin{cor}
	Over the field $\mathbb{F}_2$, there is a family of tensors $T_n$ of size
	$[n]^3$ defined for $n=2\cdot 3^l$, such that
	\begin{enumerate}
		\item $\rank(T_n)\ge 3n-\Theta(\lg n)$
		\item On inputs $n$ and $(i,j,k)\in[n]^3$, $T_n(i,j,k)$ is computable in $O(\polylog(n))$
	\end{enumerate}
\end{cor}

Thus, this algebraic construction is also explicit, at least for some values of
$n$.  Also, this corollary is not limited to $\mathbb{F}_2$.  Other
constructions~\cite{sparseirredpolys} are known over some other fields.
However, unlike the results of Section~\ref{sect:combtensors}, it is not
clear if better lower bounds exist for the tensors in this section.  Indeed, we
do not at present know non-trivial upper bounds for the tensors given here.

\section{Higher-Order Tensors}\label{sect:hightensors}

In this section we investigate order-$d$ tensors, particularly when $d$ is odd.  As
Raz~\cite{raz} shows, we can always ``reshape'' a lower-order tensor into a higher-order
tensor without decreasing rank.  Raz mentions this for reshaping an order-$d$ tensor into an
order-$2$ tensor (a matrix) and thus shows that there are explicit order-$d$ tensors with
rank $n^{\lfloor d/2\rfloor}$.  We use our results for order-$3$ tensors to derive a better
bound in the case when $d$ is odd.

We first state our reshaping lemma, keeping in mind that we again now work over an arbitrary
field.

\begin{lem}
	\label{lem:reshape}
	Let $T$ be an order-$3$ tensor of size $[n^d]\times [n^d]\times[n]$.  Then define
	the order-$(2d+1)$ tensor $T'$ of size $[n]^{2d+1}$ by
	\begin{equation*}
		T'(i_1,\ldots,i_{2d+1})=T\left(1+\sum_{j=0}^{d-1} (i_{j+1}-1) n^j,1+\sum_{j=0}^{d-1}
		(i_{j+(d+1)}-1)n^j,i_{2d+1}\right)
	\end{equation*}
	Then $T'$ has rank at least $\rank(T)$.

	Further, if $T(\cdot,\cdot,\cdot)$ is computable in time $f(n)$, then
	$T'(\cdot,\ldots,\cdot)$ is computable in time $O(\poly(d)\polylog(n)+f(n))$.
\end{lem}
\begin{proof}
	First observe that the map $(i_1,\ldots,i_d)\mapsto 1+\sum_{j=0}^{d-1} (i_{j+1}-1)
	n^j$ is a bijection from $[n]^d$ to $[n^d]$ as this is simply the base-$n$
	expansion.  That this is map is computable in time $O(\poly(d)\polylog(n))$ establishes the
	claim about efficiency.

	Now consider a rank $r'$ decomposition of $T'$
	\begin{equation*}
		T'=\sum_{l=1}^{r'} \vec{v}_{l,1}\otimes\cdots\otimes \vec{v}_{l,2d+1}
	\end{equation*}
	We now define $\vec{u}_{l,1}\in\mathbb{F}^{n^d}$.  Via the bijection from above, we
	can write
	\begin{equation*}
		\vec{u}_{l,1}\left(1+\sum_{j=0}^{d-1}(i_{j+1}-1)n^j\right)=\vec{v}_{l,1}(i_1)\cdots\vec{v}_{l,d}(i_d)
	\end{equation*}
	and similarly we define $\vec{u}_{l,2}\in\mathbb{F}^{n^d}$ by
	\begin{equation*}
		\vec{u}_{l,2}\left(1+\sum_{j=0}^{d-1}(i_{j+(d+1)}-1)n^j\right)=\vec{v}_{l,d+1}(i_{d+1})\cdots\vec{v}_{l,2d}(i_{2d})
	\end{equation*}, and we take $\vec{u}_{l,3}=\vec{v}_{l,2d+1}\in\mathbb{F}^n$.  Thus
	we see that 
	\begin{equation*}
		T=\sum_{l=1}^{r'}\vec{u}_{l,1}\otimes\vec{u}_{l,2}\otimes\vec{u}_{l,3}
	\end{equation*}
	by examining the equation pointwise, and thus $r'\ge\rank(T)$.  The conclusion thus follows when taking $r'=\rank(T')$.
\end{proof}

The above lemma shows that rank lower bounds for low-order tensors extend (weakly) to rank lower
bounds of higher-order tensors.  We now apply this lemma to the tensor rank lower bounds of
Section~\ref{sect:combtensors}.  It is possible to do similarly with the results of 
Section~\ref{sect:algtensors}, but a weaker conclusion would result as those lower bounds
are weaker.

\begin{cor}
	\label{cor:maincor-highdim}
	For every $d\ge 1$, there is a family of $\{0,1\}$-tensors $T_n$ of size $[n]^{2d+1}$ such that
	$\rank(T_n)=2n^d+n-\Theta(d\lg n)$.  Further, given $n$ and $i_1,\ldots,i_{2d+1}$,
	$T_n(i_1,\ldots,i_{2d+1})$ is computable in polynomial time (that is, in time
	$O(\poly(d)\polylog(n))$).
\end{cor}
\begin{proof}
	We first observe that the proof of Corollary~\ref{cor:maincor} extends to give a
	family of tensors $T_n$ with size
	$[n]\times [n]\times[f(n)]$, where $\rank(T_n)=2n+f(n)-\Theta(\lg n)$, where
	$\omega(\lg n)\le f(n)\le n$.  Further,
	these tensors have their entries computed in $O(\polylog(n))$ time.  
	
	Thus, this
	leads to tensors of size $[n^d]\times[n^d]\times [n]$ of rank $2n^d+n-\Theta(d\lg
	n)$.  By Lemma~\ref{lem:reshape} we can then see that these tensors can be reshaped
	into the desired tensors, establishing the claim on the rank as well as the
	explicitness.
\end{proof}

\section{Proofs for Section~\ref{sect:permtensors}}\label{sect:permtensorsproofs}

We give proofs of various claims from Section~\ref{sect:permtensors}, and
examine the tightness of some of the results.

\begin{proof}[Proof of Proposition~\ref{prop:interpolate}]

	To apply Ben-Or's interpolation  idea to tensors, we first note the connection between tensors and polynomials.
	Consider the space of polynomials
	$\mathcal{P}:=\mathbb{F}[\{X^{(1)}_i\}_{i=0}^{n-1},\ldots,\{X^{(d)}_i\}_{i=0}^{n-1}]$, that is, polynomials on
	the variables $X^{(j)}_i$ that are set-multilinear with respect to the sets
	$\{X^{(j)}_i\}_{i=0}^{n-1}$.  One can call such a polynomial \textit{simple} if it can be
	written as 
	\begin{equation*}
		\prod_{j=1}^d \left(a_{j,0}X^{(j)}_0+\cdots+a_{j,n-1}X^{(j)}_{n-1}\right)
	\end{equation*}
	One can then define the $\rank(p)$, for $p\in\mathcal{P}$, as the least
	number of simple polynomials needed to sum to $p$.  One can observe that
	$\mathcal{P}$ is a $[n]^d$ tensor product space, and the notions of rank
	coincide.  In this language, we seek to upper-bound the rank of the
	polynomial
	\begin{equation*}
		T(\{X^{(j)}_i\}_{i,j})=\sum_{m=0}^{d(n-1)}c_m\sum_{i_1+\cdots+i_d=m}X_{i_1}^{(1)}\cdots X_{i_d}^{(d)}
	\end{equation*}

	To implement Ben-Or's method, define the auxiliary polynomial $P$ by
	\begin{equation*}
		P(\alpha,\{X^{(i)}_j\}_{i,j}):=\prod_{i=1}^d \left(X^{(i)}_0+\alpha X^{(i)}_1+\alpha^2 X^{(i)}_2+\cdots+\alpha^{n-1} X^{(i)}_{n-1}\right)
	\end{equation*}
	For fixed $\alpha$, this polynomial is simple. When $\alpha$ is considered a variable, this
	polynomial has degree $d(n-1)$ in $\alpha$.  Further, the coefficient of $\alpha^m$ is
	\begin{equation*}
		C_{\alpha^m}(P(\alpha,\{X^{(j)}_i\}_{i,j}))=\sum_{i_1+\cdots+i_d=m} X^{(1)}_{i_1}\cdots X^{(d)}_{i_d}
	\end{equation*}
	which corresponds exactly to tensors of the desired form.  We can now
	interpret the auxiliary polynomials as polynomials in
	$\mathbb{F}(\{X^{(j)}_i\}_{i,j})[\alpha]$, the polynomial ring
	in the variable $\alpha$ over the field of rational functions in the
	$X^{(j)}_i$.  As $|\mathbb{F}|> d(n-1)$, we can consider the
	evaluations
	$\{P(\alpha_l,\{X{(j)}_i\}_{i,j})\}_{l=0}^{d(n-1)}\in\mathbb{F}[\{X^{(j)}_i\}_{i,j}]$ for
	distinct $\alpha_l\in\mathbb{F}$. Polynomial interpolation means that the coefficients
	$C_{\alpha^m}(P)$ are recoverable from linear combinations of the evaluations of $P$.
	As the $\alpha_l\in\mathbb{F}$, the (linear) evaluation map from the coefficients of $P$ to the
	evaluations is defined by a $\mathbb{F}$-matrix.  Therefore, the inverse of this map is also
	defined by an $\mathbb{F}$-matrix. Specifically, there are coefficients
	$a_{m,l}\in\mathbb{F}$ such that
	\begin{equation*}
		C_{\alpha^m}(P(\alpha,\{X^{(j)}_i\}_{i,j}))=\sum_{l=0}^{d(n-1)} a_{m,l} P(\alpha_l,\{X^{(j)}_i\}_{i,j})
	\end{equation*}
	Therefore,
	\begin{align*}
		T(\{X^{(j)}_i\}_{i,j})
		&=\sum_{m=0}^{d(n-1)}c_m \sum_{l=0}^{d(n-1)} a_{m,l} P(\alpha_l,\{X^{(j)}_i\}_{i,j})\\
		&=\sum_{l=0}^{d(n-1)} \left(\sum_{m=0}^{d(n-1)} c_m a_{m,l}\right) P(\alpha_l,\{X^{(j)}_i\}_{i,j})
	\end{align*}
	Thus, $T$ is in the span of $d(n-1)+1$ simple polynomials.  By moving the coefficients on
	the simple polynomials inside the product, this shows that $T$ is expressible as the sum of
	$d(n-1)+1$ simple polynomials.  Using the above connection with tensors, this shows that the
	rank is at most $d(n-1)+1$.
\end{proof}

We now prove that Proposition~\ref{prop:interpolate} is essentially tight.
The proof uses Corollary~\ref{cor:layerreduction}.

\begin{prop}
	\label{prop:interpolationistight}
	Let $\mathbb{F}$ be a field.  Let $T$ be a tensor $T:\llbracket n \rrbracket^d\to\mathbb{F}$
	such that
	\begin{equation*}
		T(i_1,\ldots,i_d)=	\begin{cases}
						1	&	\text{if }i_1+\cdots+i_d=n\\
						0	&	\text{if }i_1+\cdots+i_d>n\\
						\text{unconstrained}	&\text{else}
					\end{cases}
	\end{equation*}
	then $\rank_\mathbb{F}(T)\ge (d-1)(n-1)+1$.
\end{prop}
\begin{proof}
	The proof is by induction on $d$, using Corollary~\ref{cor:layerreduction} to achieve a lower bound.
	
	\underline{$d=1$:} As $T\ne 0$, its rank must be at least 1, so the result follows.

	\underline{$d>1$:} Decompose $T$ into layers along the $d$-th axis, so that
	$T=[T_0|\cdots|T_{n-1}]$.  Observe that the hypothesis on $T$ implies that for any linear
	combination $S=\sum_{k=0}^{n-1} c_i T_i$ it must be that $S\ne 0$.  For if not, one may consider
	the smallest $i$ such that $c_i\ne 0$.  Then $S(n-i,0,\cdots,0)=c_i$ by the hypothesis on
	$T$ and the construction of $S$, which is a
	contradiction as $c_i\ne 0$. 

	Thus, Corollary~\ref{cor:layerreduction} implies that $\rank_\mathbb{F}(T)\ge
	\rank_\mathbb{F}(T_0+\sum_{i=1}^{n-1} a_i T_i)+(n-1)$ for some $a_i\in\mathbb{F}$. However,
	observing that $T'=T_0+\sum_{i=1}^{n-1} a_i T_i$ is an order-$(d-1)$ tensor fitting the
	hypothesis of the induction we see that $\rank_\mathbb{F}(T')\ge (d-2)(n-1)+1$.  Combining
	the above equations finishes the induction.
\end{proof}

The above proposition shows that Proposition~\ref{prop:interpolate} is nearly tight, because
together they show that defining $T:\llbracket n\rrbracket^d\to\mathbb{F}$ by
$T(i_1,\ldots,i_d)=\lib i_1+\cdots+i_d=n\rib$, we see that $(d-1)(n-1)+1\le\rank_\mathbb{F}(T)\le
d(n-1)+1$.

Proposition~\ref{prop:interpolate} was done by interpolating a univariate polynomial.  By
interpolating multivariate polynomials one may obtain an upper bound for the
rank over group tensors arising from the direct product of cyclic groups.
However, the same result is derivable in a more modular fashion, which we now
present.  We start with the folklore fact that tensoring two tensors multiplies their
rank bounds.

\begin{lem}
	Let $\mathbb{F}$ be a field. Let $T:[n]^d\to \mathbb{F}$ and
	$S:[m]^d\to\mathbb{F}$ be two tensors.  Define
	$(T\otimes S):([n]\times [m])^d\to\mathbb{F}$ by
	\begin{equation*}
		(T\otimes S)((i_1,i'_1),\ldots,(i_d,i'_d))=T(i_1,\ldots,i_d)\cdot S(i'_1,\ldots,i'_d)
	\end{equation*}
	Then $\rank_\mathbb{F}(T\otimes S)\le
	\rank_\mathbb{F}(T)\rank_\mathbb{F}(S)$.
	\label{lem:tensoring-rank}
\end{lem}
\begin{proof}
	Suppose $T=\sum_{l=1}^r\otimes_{j=1}^d \vec{a}_{j,l}$ and
	$S=\sum_{l'=1}^{r'}\otimes_{j'=1}^d \vec{a}'_{j',l'}$.  Then
	$T(i_1,\ldots,i_d)=\sum_{l}\prod_j \vec{a}_{j,l}(i_j)$ and
	$S(i'_1,\ldots,i'_d)=\sum_{l'} \prod_{j'}\vec{a}'_{j',l'}(i'_{j'})$.
	Thus, $(T\otimes
	S)((i_1,i'_1),\ldots,(i_d,i'_d))$ equals $ \sum_{l,l'}\prod_{j,j'}
	\vec{a}_{j,l}(i_j) \vec{a}'_{j',l'}(i'_{j'}) =\sum_{l,l'}\prod_{(j,j')}
	(\vec{a}_{j,l}\otimes \vec{a}'_{j',l'})_{i_j,i'_{j'}}$.  Thus, as for
	fixed $l,l'$ the tensor $\prod_{(j,j')}
	(\vec{a}_{j,l}\otimes \vec{a}'_{j',l'})_{i_j,i'_{j'}}$ is simple (as a
	$([n]\times[m])^d\to\mathbb{F}$ tensor), this
	shows the claim.
\end{proof}

We now apply this to the direct product construction of groups.

\begin{cor}
	Consider integers $n_1,\ldots,n_m\in\mathbb{Z}_{\ge2}$ and consider the
	finite abelian group $G=\mathbb{Z}_{n_1}\times\cdots\times\mathbb{Z}_{n_m}$.
	Let $\mathbb{F}$ be a field with at least $\max_i(d(n_i-1)+1)$ elements. Then,
	$\rank(T_{G}^d)\le \prod_i(d(n_i-1)+1)$.
	\label{cor:interpolate-prod-of-cycle}
\end{cor}
\begin{proof}
	First observe the relevant definitions imply that $T_{\mathbb{Z}_n}^d\otimes
	T_{\mathbb{Z}_m}^d=T_{\mathbb{Z}_n\times\mathbb{Z}_m}^d$.  Thus, the
	claim follows directly from Corollary~\ref{cor:interpolate-cyclic} and
	Lemma~\ref{lem:tensoring-rank}.
\end{proof}

We now recall the Structure Theorem of Finite Abelian Groups.

\begin{thm}[Structure Theorem of Finite Abelian Groups (see, e.g.\
\cite{artin})]
	Let $G$ be a finite abelian group.  Then there are (not necessarily
	distinct) prime powers $n_1,\ldots,n_m\in\mathbb{Z}_{\ge 2}$ such that
	$G=\mathbb{Z}_{n_1}\times\cdots\times\mathbb{Z}_{n_m}$.
\end{thm}

This theorem shows that Corollary~\ref{cor:interpolate-prod-of-cycle} extends to general
groups.  One can get better bounds if more information is known about the group,
of if results such as Theorem~\ref{thm:repthyrank} apply, but the next result
shows that even without such information group tensors from finite abelian
groups have ``low'' rank.

\begin{proof}[Proof of Corollary~\ref{cor:interpolate-abel-group}]
	Observe that using the Structure Theorem of Finite Abelian groups, we
	can apply Corollary~\ref{cor:interpolate-prod-of-cycle} to
	$G=\prod_{i=1}^m
	\mathbb{Z}_{n_i}$, and
	using that $d(n_i-1)+1\le dn$ (as $d\ge 2$) shows that
	$\rank_\mathbb{F}(T_G^d)\le
	d^m\prod n_i$.  As $m\le \lg|G|$ and $|G|=\prod n_i$, the result
	follows.
\end{proof}

\begin{proof}[Proof of Lemma~\ref{lem:changefields}]
	Define $m:=\dim_\mathbb{F}\mathbb{K}$. Thus, we can identify
	$\mathbb{K}=\mathbb{F}^m$ as vector spaces, where we choose that
	$1\in\mathbb{K}$ is the first element in the $\mathbb{F}$-basis for
	$\mathbb{K}$.  This gives an $\mathbb{F}$-algebra
	isomorphism $\mu$ between $\mathbb{K}$ and a sub-ring $M$ of the $m\times m$
	$\mathbb{F}$-matrices, where $M$ is defined as the image of $\mu$.  The
	map $\mu$ is defined by associating
	$x\in\mathbb{K}$ with the matrix inducing the linear map
	$\mu(x):\mathbb{F}^m\to\mathbb{F}^m$, where $\mu(x)$ is the multiplication
	map of $x$.  That is, using that $\mathbb{K}=\mathbb{F}^m$ we can see
	that the map $y\mapsto xy$ for $y\in\mathbb{K}$ is an
	$\mathbb{F}$-linear map, and thus defines $\mu(x)$ over $\mathbb{F}^m$.

	That the map is injective follows from the fact that $\mu(x)$ must map
	$1\in\mathbb{K}$ to $x\in\mathbb{K}$, so $x$ is recoverable from
	$\mu(x)$ (and surjectivity follows be definition of $M$).  To see the
	required homomorphism properties is also not difficult.  As
	$(x+z)y=xy+zy$ for any $x,y,z\in\mathbb{K}$, this shows that
	$\mu(x+z)=\mu(x)+\mu(z)$ as linear maps, and thus as matrices.
	Similarly, as $(xz)y=x(zy)$ for any $x,y,z\in\mathbb{K}$ it must be that
	$\mu(xz)=\mu(x)\mu(z)$.  That this map interacts linearly in
	$\mathbb{F}$ implies that it is an $\mathbb{F}$-algebra homomorphism, as
	desired.

	Now consider a tensor $T:[n]^d\to\mathbb{F}$ with simple tensor
	decomposition $T=\sum_{l=1}^{\rank_\mathbb{K}(T)} \otimes_{j=1}^d \vec{a}_{j,l}$ over
	$\mathbb{K}$.  First observe that if we define the map
	$\pi:\mathbb{K}\to\mathbb{F}$ defined by
	\begin{equation*}
		\pi(x)=\begin{cases}
			x	&\text{if }x\in\mathbb{F}\\
			0	&\text{else}
			\end{cases}
	\end{equation*}
	then $T=\sum_{l=1}^{\rank_\mathbb{K}(T)} \pi(\otimes_{j=1}^d \vec{a}_{j,l})$.  Thus for each $l$,
	$\pi(\otimes_{j=1}^d \vec{a}_{j,l})$ is a tensor
	$T_l:[n]^d\to\mathbb{F}$.
	
	We now show
	that $\rank_\mathbb{F} (T_l)\le m^{d-1}$.  First observe that for
	$x\in\mathbb{F}$, $\mu(x)$ is a diagonal matrix.  In particular, because
	we chose $1\in\mathbb{K}$ to the first element in the $\mathbb{F}$-basis for $\mathbb{K}$,
	for $x\in\mathbb{F}$, $\pi(x)$ is equal to the
	$(1,1)$-th entry in $\mu(x)$.  Thus, it follows that
	$T_l(i_1,\ldots,i_d)=\left(\mu(\vec{a}_{1,l}(i_1))\cdots
	\mu(\vec{a}_{d,l}(i_d))\right)_{1,1}$.  
	By expanding out the matrix multiplication
	we can see that $T_l$ is
	expressible as
	\begin{equation*}
		T_l(i_1,\ldots,i_d)= \sum_{k_1=1}^m\cdots\sum_{k_{d-1}=1}^m
		\mu(\vec{a}_{1,l}(i_1))_{1,k_1}\cdot
		\mu(\vec{a}_{2,l}(i_2))_{k_1,k_2}\cdots
		\mu(\vec{a}_{d-1,l}(i_{d-1}))_{k_{d-2},k_{d-1}}\cdot
		\mu(\vec{a}_{d,l}(i_d))_{k_{d-1},1}
	\end{equation*}
	and just as in Theorem~\ref{thm:repthyrank} we see that for fixed $k_j$
	the summands are simple $\mathbb{F}$-tensors, and thus
	$\rank_{\mathbb{F}}(T_l)\le m^{d-1}$.
	
	Using the observation that $T=\sum_{l=1}^{\rank_\mathbb{K}(T)} T_l$
	and the above bound for the $\mathbb{F}$-rank of $T_l$, we then see that
	$\rank_\mathbb{F}(T)\le
	(\dim_\mathbb{F}\mathbb{K})^{d-1}\rank_\mathbb{K}(T)$, as desired.
\end{proof}

\section{Proofs of Section~\ref{sect:monotone}}\label{sect:monotoneproofs}

\begin{proof}[Proof of Theorem~\ref{thm:monotone}]
	\underline{$\mrank_\mathbb{F}(T_{\mathbb{Z}_n}^d)\ge n^{d-1}$:} We remark that the following
	lower bound will only rely on the fact that $T_{\mathbb{Z}_n}^d$ is a permutation tensor,
	and no other properties.  
	
	In monotone
	computation, there is no cancellation of terms. Thus, in a monotone simple tensor
	decomposition $T=\sum_{l=1}^r T_l$, one can see that the partial sums $T_{\le
	m}=\sum_{l=1}^m T_l$ successively cover more and more of the non-zero entries of $T$.  We
	will show that in any monotone decomposition of $T_{\mathbb{Z}_n}^d$, at most one non-zero
	entry can be covered by any $T_l$, which implies that the monotone rank is at least the
	number of non-zero entries, which is $n^{d-1}$.

	We now prove that in any monotone simple tensor decomposition
	$T_{\mathbb{Z}_n}^d=\sum_{l=1}^r\otimes_{j=1}^d \vec{a}_{j,l}$, each simple tensor
	$T_l:=\otimes_{j=1}^d \vec{a}_{j,l}$ can cover at most one non-zero entry of
	$T_{\mathbb{Z}_n}^d$.  Suppose not, for contradiction.  Then there is a simple tensor $T_l$
	that covers at least two non-zero entries $(i_1,\ldots,i_d)$ and $(i'_1,\ldots,i'_d)$ of
	$T_{\mathbb{Z}_n}^d$.  However, these tuples must differ in at least one index, we we assume
	without loss of generality to be index 1, so that $i_1\ne i'_1$.  Consequently, we must have that
	$\vec{a}_{1,l}(i_1),\vec{a}_{1,l}(i'_1)> 0$ (as all field constants are positive in monotone
	computation).  As
	$\vec{a}_{j,l}(i_j)> 0$ for $j>1$ (as
	$T_l(i_1,i_2,\ldots,i_d)=1$), it must be that $T_l(i'_1,i_2,\ldots,i_d)=\vec{a}_{1,l}(i'_1)
	\prod_{j>1}^d\vec{a}_{j,l}(i_j)> 0$.  However, this is a contradiction.  For now this
	positive number at $T_l(i'_1,i_2,\ldots,i_d)$ cannot be canceled out by other simple tensors
	in a monotone computation and we must have $T_{\mathbb{Z}_n}^d(i'_1,i_2,\ldots,i_d)=0$ by
	the fact that this is a permutation tensor.  Thus, it must be that each simple tensor in
	this monotone computation can only cover a single non-zero entry of $T_{\mathbb{Z}_n}^d$,
	which implies the lower bound by the above argument.
\end{proof}

\end{document}